\documentclass[10pt]{article}

\usepackage{amssymb}
\usepackage{amsmath, amsthm}
\usepackage{epsfig, enumerate}
\usepackage{mathdots}
\usepackage[utf8]{inputenc}

\newtheorem{theo}{Theorem}
\newtheorem{lemma}[theo]{Lemma}
\newtheorem{prop}[theo]{Proposition}
\newtheorem{coro}[theo]{Corollary}
\newtheorem{remark}[theo]{Remark}
\newtheorem*{rem}{Remark}
\newtheorem*{define}{Definition}

\newtheorem*{assumptions}{Assumptions}

\newtheorem*{question}{Question}

\newcommand{\CC}{{\mathbb C}}
\newcommand{\NN}{{\mathbb N}}
\newcommand{\RR}{{\mathbb R}}

\newcommand{\ZZ}{{\mathbb Z}}

\newcommand{\GG}{{\mathbb G}}

\newcommand{\PP}{{\mathbb P}}

\newcommand{\FF}{{\mathbb F}}
\newcommand{\Aa}{{\mathcal A}}
\newcommand{\Pp}{{\mathcal P}}

\newcommand{\EE}{{\mathbb E}}
\newcommand{\Bb}{{\mathcal B}}
\newcommand{\Dd}{{\mathcal D}}
\newcommand{\Ff}{{\mathcal F}}
\newcommand{\Ee}{{\mathcal E}}
\newcommand{\Gg}{{\mathcal G}}
\newcommand{\Ww}{{\mathcal W}}
\newcommand{\Vv}{{\mathcal V}}
\newcommand{\Ss}{{\mathcal S}}

\newcommand{\Oo}{{\mathcal O}}
\newcommand{\Tr}{\mbox{\rm Tr}}
\newcommand{\Tt}{{\mathcal T}}
\newcommand{\Rr}{{\mathcal R}}
\newcommand{\Nn}{{\mathcal N}}

\newcommand{\Cc}{{\mathcal C}}
\newcommand{\Jj}{{\mathcal J}}

\newcommand{\Qq}{{\mathcal Q}}

\newcommand{\Yy}{{\mathcal Y}}
\newcommand{\Xx}{{\mathcal X}}

\newcommand{\one}{{\bf 1}}
\newcommand{\nul}{{\bf 0}}

\newcommand{\diag}{{\rm diag}}

\newcommand{\GL}{{\rm GL}}
\newcommand{\Mat}{{\rm Mat}}

\newcommand{\HSp}{{\rm HSp}}

\newcommand{\bgmat}{\left(\begin{matrix}}
\newcommand{\fnmat}{\end{matrix}\right)}

\newcommand{\pmat}[1]{\begin{pmatrix} #1 \end{pmatrix}}
\newcommand{\smat}[1]{\left(\begin{smallmatrix} #1 \end{smallmatrix}\right)}

\newcommand{\wt}{\widetilde}
\newcommand{\wh}{\widehat}
\newcommand{\ov}{\overline}

\newcommand{\zb}{\mathbf{z}}

\DeclareMathOperator{\re}{{\rm Re}}
\DeclareMathOperator{\im}{{\rm Im}}

\newcommand{\qtx}[1]{\quad \text{#1} \quad}

\numberwithin{equation}{section}
\numberwithin{theo}{section}

\title{A central limit theorem for products of random matrices
and GOE statistics for the Anderson model on long boxes}
\author{Christian Sadel$^1$\thanks{The work of C.S. was supported by NSERC Discovery grant 92997-2010 RGPIN and 
by the People Programme (Marie Curie Actions) of the EU 7th Framework Programme FP7/2007-2013, REA grant 291734.} , B\'alint Vir\'ag$^2$\\[.2cm]
{\small $^1$Department of Mathematics, University of British Columbia, Vancouver, BC V6T 1Z2, Canada} \\
{\small $^1$Institute of Science and Technology, 3400 Klosterneuburg, Austria} \\
{\small $^2$Department Mathematics, University of Toronto, ON M5S 2E4, Canada}}

\begin{document}

\maketitle

\setcounter{tocdepth}{2}

\begin{abstract}
 We consider products of random matrices that are small, independent identically distributed
 perturbations of a fixed matrix $\Tt_0$.
 Focusing on the eigenvalues of $\Tt_0$ of a particular size we obtain a limit to a SDE in a critical scaling.
 Previous results required $\Tt_0$ to be a (conjugated) unitary matrix so it could not have eigenvalues of different modulus.
 From the result we can also obtain a limit SDE for the Markov process given by the action of the random products on the flag manifold.
 Applying the result to random Schr\"odinger operators we can improve some result by Valko and Virag showing GOE statistics for the rescaled eigenvalue process of a sequence of Anderson models on long boxes.
 In particular we solve a problem posed in their work.
\end{abstract}

\tableofcontents

\section{Introduction and results \label{sec-intro}}

The goal of this paper is to study scaling limits of random matrix products
$$\Tt_{\lambda,n} \Tt_{\lambda,n-1}\cdots \Tt_{\lambda,1}$$
with $\lambda\to 0$ where the $\Tt_{\lambda,n}$ are perturbations of a fixed $d\times d$ matrix of the form
\begin{equation} \label{eq-def-Tt}
  \Tt_{\lambda,n} = \Tt_0 + \lambda \Vv_{\lambda,n} + \lambda^2 \Ww_\lambda\;.
\end{equation}
Here, for every $\lambda$, $\Vv_{\lambda,n}$ is a family of
independent, identically distributed random matrices
with $\EE(\Vv_{\lambda,n})=\nul$ and
$\Ww_{\lambda}$ is a deterministic matrix, both of order one.
In the simplest case, $d=1$, $\Tt_0=1$, $\Ww_\lambda=0$, and $\Vv_{\lambda,n}=\Vv_n$ are independent centered random variables
with variance one. Then, the classical Donsker's central limit theorem applied to the logarithm of $(1)$ shows that the product (for here denoted $\Xx_{\lambda, n}=\Tt_{\lambda,n}\Tt_{\lambda,n-1}\cdots \Tt_{\lambda,1}$) satisfies
$$
(\Xx_{\lambda;t/\lambda^2}, t\ge 0)\; \Longrightarrow \;(e^{B(t)-t/2},\; t\ge 0)
$$
as $\lambda \to 0$, where $B(t)$ is a standard Brownian motion. Compared to the simplest case, the general case has extra interesting features.
\begin{itemize}
\item The matrix $\Tt_0$ can have eigenvalues of different absolute values, so the product can grow exponentially at different rates in different directions.
\item The matrix $\Tt_0$ can have complex eigenvalues of the same absolute value that act as rotations; this can produce an averaging effect for the added drift and noise terms.
\end{itemize}

The main question that we resolve in this paper is the following.
\begin{question}
If the matrix $\Tt_0$ have an eigenvalue of large absolute value, can one still understand the fine evolution of the product in the directions belonging to smaller eigenvalues?
\end{question}
Matrix products of
this kind are used in the study of quasi-1-dimensional random Schr\"odinger operators, and the large eigenvalues are related to so-called hyperbolic channels. Indeed, the main motivating example is this case, which we will introduce in Section~\ref{sub-RS}.
When $\Tt_0$ is (the multiple of) a unitary matrix this type of result has been
established in that context, \cite{BR,SS2,VV1,VV2} and the limiting process is described by a
stochastic differential equation (SDE).
In \cite{KVV,VV1} the SDE limit was used to study the limiting eigenvalue statistics of such random
Schr\"odinger operators in a critical scaling $\lambda^2 n=t$
(at bandedges one has a different scaling as mentioned in Appendix~\ref{sub-Jordan}).
We can extend this result and obtain a limit for the rescaled eigenvalue process
in the presence of hyperbolic channels as well (cf. Theorem~\ref{theo-EV}).
In particular, we solve Problem~3 raised in \cite{VV1} and obtain limiting GOE statistics
for the Anderson model on sequences of long boxes (cf. Theorem~\ref{theo-GOE}) with appropriate scalings.
We essentially reduce the proof to a situation where it is left to analyze the same family of SDEs as in \cite{VV1}. Deriving the GOE statistics then
relies on the work of Erd\"os, Schlein, Yau and Yin \cite{ESYY, EYY}, but we do not repeat 
these steps that are done in \cite{VV1}. The main result is stated in Section~\ref{sub-RS}, further details are given in Section~\ref{sec-RSO}.

In the considered scaling limit $\lambda^2 n=t$ in Theorem~\ref{theo-EV} localization effects and Poisson statistics are not seen and the description through an SDE arises. 
We obtain that the hyperbolic channels only shift the eigenvalues but do not affect the local statistics.
In fact, fixing the width and base energy, the local eigenvalue statistics only depends on the number of so called elliptic channels. This can be seen as some universality statement.
Increasing the number of elliptic channels and choosing appropriate sequences of models, the GOE statistics arises.

As a byproduct of this work we solve some conjecture from \cite{Sa1} showing that
there is an SDE limit for the reduced transfer matrices in the presence of hyperbolic channels (cf. Remark~\ref{rem-red-T}) .

\vspace{.2cm}

Random matrix ensembles such as the Gaussian Orthogonal Ensemble (GOE) were introduced by Wigner \cite{Wi} to model the observed repulsion between eigenvalues in large nuclei. 
The local statistics is given by the ${\rm Sine}_1$ kernel, see e.g. \cite{Me}. 
This type of repulsion statistics is expected for many randomly disordered systems of the same symmetry class (time reversal symmetry).
that have delocalized eigenfunctions. This is referred to as {\bf universality}. 
Most models with rigorously proved universal bulk behavior are themselves ensembles of random matrices, e.g. \cite{DG, ESY, Joh, TV}.
Recently, T. Shcherbina proved universal GUE statistics (Gaussian Unitary Ensemble) for random block band matrix ensembles that 
in some sense interpolate between the classical matrix ensembles and Anderson models \cite{Shc}.

The Anderson model was introduced by P. W. Anderson to describe disordered media like randomly doped semi-conductors
\cite{And}. It is given by the Laplacian and a random independent identically distributed potential and has significantly less randomness than the matrix ensemble models.
For large disorder and at the edge of the spectrum, the Anderson model in $\ZZ^d$ or $\RR^d$ 
localizes \cite{FS, DLS, SW, CKM, AM, Aiz, Klo}
and has Poisson statistics \cite{Mi, Wa, GK}. 
For small disorder in the bulk of the spectrum, localization and Poisson statistics appears in one and quasi-one dimensional systems \cite{GMP, KuS, CKM, Lac, KLS} (except if prevented by a symmetry \cite{SS3})
and is expected (but not proved) in 2 dimensions.
Delocalization for the Anderson model was first rigorously proved on regular trees (Bethe lattices) \cite{K1} and had been extended to several 
infinite-dimensional tree-like graphs \cite{K1, ASW, FHS, AW, KLW, FHH, KS, Sa2, Sa3, Sha}.
Recently it was shown that there is a transition from localization to delocalization on normalized antitrees at exactly 2-dimensional growth rate \cite{Sa4}.
For 3 and higher dimensions one expects delocalized eigenfunctions (absolutely continuous spectrum) for small disorder and the eigenvalue statistics of large boxes should approximate GOE by universality.
However, proving any of these statements in $\ZZ^d$ or $\RR^d$ for $d\geq 3$ remains a great mathematical challenge.

\vspace{.2cm}

The papers \cite{BR, VV1, VV2} are restricted to the subset of the important cases where all eigenvalues of $\Tt_0$ had the same absolute value (and no Jordan blocks). 
The novelty of this work is to handle eigenvalues of different absolute value for $\Tt_0$, the application to Schr\"odinger operators 
comes from applying Theorem~\ref{theo-SDE} to the transfer matrices. 

Let us briefly explain why this is not a trivial extension.
If $\Tt_0$ (or $A\Tt_0 A^{-1}$ for some matrix $A$) is unitary one simply has to remove the free evolution from the random products. 
To illustrate this, let for now $\Xx_{\lambda,n}=\Tt_{\lambda,n} \Tt_{\lambda,n-1}\cdots \Tt_{\lambda,1}$. Then, 
$\Tt_0^{-n} \Xx_{\lambda,n}=(\one+\lambda \Tt_0^{-n} \Vv_{\lambda,n} \Tt_0^{(n-1)} + \lambda^2 \Tt_0^{-n} \Ww_\lambda \Tt_0^{(n-1)}) (\Tt_0^{-(n-1)}\Xx_{\lambda,n-1})$.
The conjugations like $\Tt_0^{-n} (\Vv_{\lambda,n}\Tt_0^{-1}) \Tt_0^n$ simply lead to an averaging effect over the compact group generated by the unitary $\Tt_0$ in the limit for the drift and diffusion terms.
Adopting techniques by Strook and Varadhan \cite{SV} and Ethier and Kurtz \cite{EK} to this situation (cf. Proposition~23 in \cite{VV2}) 
one directly obtains an SDE limit for $\Tt_0^{-n} \Xx_{\lambda,n}$ in the scaling $\lambda^2 n= t$.
If $\Tt_0$ has eigenvalues of different sizes then generically some entries of $\Tt_0^{-n} \Ww_\lambda \Tt^{n-1}$ and the variance of some entries of
$\Tt_0^n \Vv_{\lambda,n}\Tt_0^{n-1}$ will grow exponentially in $n$. This destroys any hope of a limiting process.
Instead one may then consider a process $U^{-n} \Xx_{\lambda,n}$  where $U$ is a unitary just counteracting the fast rotations. 
But then one still has different directions growing at different exponential rates 
even for the free evolution, and simply projecting to some subspace, $PU^{-n}\Tt_0^{-n} \Xx_{\lambda,n}$, does not work either!

The trick lies in finding a projection which cuts off the exponential growth of the free evolution {\bf and} does not screw up the convergence of the
random evolution to some drift and diffusion terms. 
The correct way to handle the exponential growing directions is choosing a Schur complement. The exponentially decreasing directions will tend to zero and not matter and the directions of size $1$
will lead to a limit. The exponential growing directions have some non-trivial effect and lead to an additional drift term. 
As the Schur complement itself is not a Markov process, it will be better to consider it as part of a quotient of $\Xx_{\lambda,n}$ modulo a certain subgroup of $\GL(d,\CC)$.
Then one still needs several estimates to handle the appearing inverses in the Schur complement and the error terms before one can apply Proposition~\ref{p_turboEK}
which is some modification of Proposition~23 in \cite{VV2}. 

\vspace{.2cm}

In the Appendix we will state some further general corollaries or applications of Theorem~\ref{theo-SDE}. 
Although we cannot take an SDE limit of the entire
matrix as indicated above, it will be possible to describe the limit of its action on Grassmannians and flag manifolds (cf. Appendix~\ref{sec-corr}). 
For this limit the correlations for SDE limits corresponding to different sizes of eigenvalues of $\Tt_0$ are important.
The limit processes live in certain submanifolds that are stable under the free, non-random
dynamics of $\Tt_0$. This result is related to the numerical calculations in \cite{RS} who considered the action of the transfer matrices on the so called Lagrangian Planes, or Lagrangian Grassmannians 
(which is some invariant subspace of a Grassmannian).
The limiting submanifold corresponds to the 'freezing' of some phases related to the hyperbolic channels. In the scaling limit, only a motion of the part corresponding to the so called elliptic channels can be seen
and it is described by a SDE.

In Appendix~\ref{sub-Jordan} we also study the case of non-diagonalizable Jordan blocks.  These can be dealt with by a $\lambda$-dependent basis change
which leads to a different critical scaling. In the Schr\"odinger case such Jordan blocks appear at band-edges and
we give an example for a Jordan block of size $2d$ for general $d$.


\subsection{General SDE limits \label{sub-results}}

Without loss of generality we focus on the eigenvalues with absolute value $1$ and assume that $\Tt_0$ has no Jordan blocks for eigenvalues of size $1$.
Next, we conjugate the matrices $\Tt_{\lambda,n}$ to get $\Tt_0$ in Jordan form. We may write it as a block diagonal matrix of dimension $d_0+d_1+d_2$ of the form
\begin{equation}\label{eq-T_0}
\Tt_0=
\begin{pmatrix}
 \Gamma_0 &  &  \\  &     U & \\ & &  \Gamma_2^{-1}
\end{pmatrix}\;
\end{equation}
where $U$ is a unitary, and $\Gamma_0$ and $\Gamma_2$ have spectral radius smaller than $1$.
The block $\Gamma_0$ corresponds to the exponential decaying directions and the block $\Gamma_2^{-1}$ to the
exponential growing directions of $\Tt_0$.

The only way the matrix product $\Tt_{\lambda,n} \cdots \Tt_{\lambda,1} $ can have a continuous limiting evolution is if we compensate for the macroscopic rotations
given by $U$ (as in \cite{BR,VV1,VV2}).
Hence define
\begin{equation} \label{eq-def-Xx}
 \Xx_{\lambda,n} := \Rr^{-n} \Tt_{\lambda,n}\Tt_{\lambda,n-1} \cdots \Tt_{\lambda,1} \Xx_0
 \quad\text{where}\quad
 \Rr=\pmat{\one_{d_0} & &  \\ & U &  \\  &  & \one_{d_2}}\,
 \end{equation}
where $\Xx_0$ is some initial condition and  $\one_d$ is the identity matrix of dimension $d$.

In most of the following calculations
we will use a subdivision in blocks of size
$d_0+d_1$ and $d_2$. Let us define the Schur complement $X_{\lambda,n}$
of size $(d_0+d_1) \times (d_0+d_1)$ by
\begin{equation}
X_{\lambda,n} = A_{\lambda,n}-B_{\lambda,n}D_{\lambda,n}^{-1} C_{\lambda,n}\,,\quad
\text{where}\quad
 \Xx_{\lambda,n}=\begin{pmatrix}
                  A_{\lambda,n} & B_{\lambda,n} \\ C_{\lambda,n} & D_{\lambda,n}
                 \end{pmatrix}
\label{eq-def-X}\,.
\end{equation}
If $\Xx_{\lambda,n}$ and $D_{\lambda,n}$ are both invertible, then
\begin{equation} \label{eq-def-X-2}
 X_{\lambda,n} = (\Pp^* \Xx_{\lambda,n}^{-1} \Pp)^{-1}\quad
\end{equation}
where $\mathcal P$ is the projection to the first $d_0+d_1$ coordinates.
Note that invertibility of $D_{\lambda,n}$ is required to define $X_{\lambda,n}$.
Therefore, we demand the starting value $D_0$ to be invertible, where
\begin{equation*}
 \Xx_0=\begin{pmatrix}
        A_0 & B_0 \\ C_0 & D_0
       \end{pmatrix}\;,\quad\text{and we define}\quad
 X_0=A_0-B_0 D_0^{-1} C_0\,.
\end{equation*}

The first important observation, explained in Section~\ref{sec-evol}, is that the pair
\begin{equation}\label{eq-def-Z}
 (X_{\lambda,n}, Z_{\lambda,n})\quad\text{where}\quad
 Z_{\lambda,n}=B_{\lambda,n} D_{\lambda,n}^{-1}
\end{equation}
is a Markov process. Therefore, it will be more convenient to study this pair.

\vspace{.2cm}

We need the following assumptions.
\begin{assumptions} \label{assumption}
We assume that for
some constants $\epsilon>0$, $\lambda_0>0$ one has
\begin{align}
\label{eq-cond-Vv-m}
 &\sup_{0<\lambda< \lambda_0} \EE(\|\Vv_{\lambda,n}\|^{6+\epsilon}) < \infty\;.
\end{align}
Furthermore we assume that the limits of first and second moments
\begin{align}
\label{eq-cond-Vv-lim}
 \lim_{\lambda\to 0} \Ww_{\lambda}\;,\qquad
 \lim_{\lambda\to 0} \EE((\Vv_{\lambda,n})_{i,j} (\Vv_{\lambda,n})_{k,l}) \;, \qquad
 &\lim_{\lambda\to 0} \EE((\Vv^*_{\lambda,n})_{i,j} (\Vv_{\lambda,n})_{l,k}) \qquad \text{exist.}
\end{align}
\end{assumptions}

\newcommand{\vl}[2]{V_{\lambda,{#1}{#2} }}
\newcommand{\wl}[2]{W_{\lambda,{#1}{#2} }}

\vspace{.2cm}

In order to state the main theorem, we need to subdivide $\Vv_{\lambda,n},\,\Ww_\lambda$ in blocks of sizes $d_0,\, d_1,\,d_2$.
We denote the $d_j \times d_k$ blocks by
\begin{equation}\label{eq-def-Vv}
\vl{j}{k}, \qquad \text{and} \qquad \wl{j}{k} \qquad \text{respectively.}
\end{equation}
The covariances of the $d_1\times d_1$ block $\vl11$ will be important.
A useful way to encode covariances of centered matrix-valued random variables $A$ and $B$ is
to consider the matrix-valued linear functions
$M\mapsto \EE(A^\top M B)$ and $M\mapsto \EE(A^* M B)$. Choosing matrices $M$ with one entry $1$ and all other entries zero one can read off $\EE(A_{ij} B_{kl})$ and $\EE(\ov{A}_{ij} B_{kl})$ directly.
Let us therefore define
\begin{equation}\label{eq:h}
 h(M):=\lim_{\lambda\to 0} \EE(\vl11^\top M \vl11)\;,\qquad
 \wh{h}(M):=\lim_{\lambda\to 0} \EE(\vl11^* M \vl11)\;.
\end{equation}
Furthermore the lowest order drift term of the limit will come from the lowest order Schur-complement and hence contain some influence from
the exponentially growing directions.
Therefore, let
\begin{equation}\label{eq-def-W}
 W:=\lim_{\lambda\to 0} \wl11-\EE( \vl12\Gamma_2\vl21)\;.
\end{equation}
By the assumption \eqref{eq-cond-Vv-lim} above these limits exist.

\begin{theo}\label{theo-SDE}
 Let the assumptions \eqref{eq-cond-Vv-m} and \eqref{eq-cond-Vv-lim} stand.
 Then, for $t>0$ we have convergence in law,
 $Z_{\frac1{\sqrt{n}},\lfloor tn \rfloor} \; \Longrightarrow \; \nul$ and
 \begin{equation}\notag
 X_{\frac{1}{\sqrt{n}},\lfloor tn \rfloor} \; \Longrightarrow \; X_t=
 \pmat{\nul_{d_0} & \\  & \Lambda_t} X_0 \qtx{for} n\to\infty\;.
 \end{equation}
 $\Lambda_t$ is a $d_1\times d_1$ matrix valued process and the solution of
\begin{equation}\notag
 d\Lambda_{t}=V \Lambda_t \,dt + d\Bb_t\,\Lambda_t\,,\quad
 \Lambda_0={\bf 1}\,.
\end{equation}
and $\Bb_t$ is a complex matrix Brownian motion (i.e.\ $\Bb_t$ is Gaussian) with covariances
\begin{equation} \label{eq-variances}
\EE(\Bb_t^\top M \Bb_t)=
g(M) t\,,\quad
\EE(\Bb_t^* M \Bb_t)=\wh g(M) t
\end{equation}
where
\begin{align} \label{eq-def-V}
 V &=\int_{\langle U\rangle}
   u W U^*  u^* \,d u  \\
 \label{eq-def-g}
 g(M) &= \int_{\langle U\rangle} \ov u \,\ov U \,h( u ^\top M  u )\,U^*  u ^*\,d u  \;\\
 \label{eq-def-hg}
 \wh g(M) &=
 \int_{\langle U\rangle}  u  U\, \wh h( u ^* M  u )\, U^*  u ^*\,d u \;.
\end{align}
Here, $\langle U\rangle$ denotes the compact abelean group generated by the unitary $U$,
i.e. the closure of the set of all powers of $U$, and $d u $ denotes the Haar measure on $\langle U \rangle$.
\end{theo}

\begin{rem} \label{rem-SDE}
\begin{enumerate}[{\rm (i)}]
\item The analogous theorem in the situation $d_2=0$ (no exponential growing directions)
holds. In this case the matrices $B_{\lambda,n}, C_{\lambda,n}$ and $D_{\lambda,n}$ do not
exist and one simply has $X_{\lambda,n}=A_{\lambda,n}=\Xx_{\lambda,n}$. For this case one can actually simplify
some of the estimates done for the proof, as one does not need to work with the process
$B_{\lambda,n}D_{\lambda,n}^{-1}$ and no inverse is required.
\item In the case where $d_0=0$, i.e. no exponential decaying directions, the Theorem
  also works fine. In this case one simply has $X_t=\Lambda_t X_{11}$.
\item The Theorem does not hold for $t=0$ and indeed it looks contradictory for small $t$.
However, the exponentially decaying directions go to zero exponentially fast so that one obtains
\begin{equation*}
 X_{\frac1{\sqrt{n}},\lfloor n^\alpha \rfloor} \,\Longrightarrow\, \pmat{\nul & \\  & \one_{d_1}} \,X_0
\end{equation*}
for sufficiently small $\alpha$ which gives the initial conditions for the limiting process.
\item \label{rem1-rot} When defining the process $\Xx_{\lambda,n}$
one may want to subtract some of the oscillating terms in the growing and decaying directions as well, i.e.
one may want to replace $\Rr$ in \eqref{eq-def-Xx} by a unitary of the form
$\hat\Rr=\smat{U_0&&\\\ &U&\\\ &&U_2}$ written in blocks of sizes $d_0, d_1, d_2$, respectively.
Then let $$\hat\Xx_{\lambda,n}=\hat\Rr^{-n}\Rr^n \Xx_{\lambda,n}=
\pmat{\hat A_{\lambda,n} & \hat B_{\lambda,n} \\ \hat C_{\lambda,n} & \hat D_{\lambda,n}}$$
and define
the corresponding Schur complement $\hat X_{\lambda,n}=\hat A_{\lambda,n} - \hat B_{\lambda,n}
\hat D_{\lambda,n}^{-1}\hat C_{\lambda,n} $ as well has $\hat Z_{\lambda,n}= \hat B_{\lambda,n} \hat D_{\lambda,n}^{-1}$.
Simple algebra shows that
$$
\hat X_{\lambda,n} = \pmat{U_0^{-n} & \\  & \one} X_{\lambda,n}\;,\qquad
 \hat Z_{\lambda,n} = \pmat{U_0^{-n} & \\  & \one} Z_{\lambda,n} U_2^n\;.
$$
Hence, it is easy to see that for $n\to \infty$,
$$
\hat Z_{\frac1{\sqrt{n}},\lfloor nt \rfloor} \Rightarrow 0\;,\qquad
\hat X_{\frac1{\sqrt{n}},\lfloor nt \rfloor} \Rightarrow X_t
$$
where $X_t$ is the exact same process as in Theorem~\ref{theo-SDE}.
\item When $\Tt_0$ has eigenvalues of absolute value $c$ different from $1$, and $\Tt_0$ is diagonalized (or in Jordan form)
so that the corresponding eigenspace are also the span of coordinate vectors and have no Jordan blocks, then
we can apply Theorem~\ref{theo-SDE} to the products of $\Tt_{\lambda,n}/c$. Moreover, considering products of direct sums 
$\Tt_{\lambda,n}/c \oplus \Tt_{\lambda,n}/c'$ the application of Theorem~\ref{theo-SDE} gives correlations of these SDEs, cf. Theorem~\ref{theo-SDE2}.

\item The stated SDE limit can be seen as limiting processes of equivalence classes of $\Rr^{-n} \Xx_{\lambda,n}$ modulo a certain group $\GG$.
This means the following: For a specific subgroup $\GG\subset \GL(d,\CC)$ we define two matrices to be equivalent, $M\sim M'$, if and only if
$M=M' G$ for some $G\in\GG$. This equivalence relation defines the quotient $\Mat(d,\CC)/\GG$. 
One may ask for which subgroups $\GG$ of $\GL(d,\CC)$ one has some normalization $\widetilde\Rr$ such that the process 
$\widetilde\Rr^{-n} \Tt_{1/\sqrt{n}, n} \cdots \Tt_{1/\sqrt{n},1}\,\Xx_{0} \;/\; \GG$
has a distributional limit. 
As we will show in Appendix~\ref{sec-corr}, for invertible and diagonalizable $\Tt_0$
we get such a limit on the so-called flag manifold (cf. Theorem~\ref{th-flag}) and whenever $\GG$ is algebraic and cocompact (cf.  Corollary~\ref{cor-flag}).

\item Finally, one might pose the question whether one can also obtain some similar result in the presence of Jordan blocks.
In fact, combining this work with techniques of \cite{SS1} one can obtain a limit for a process obtained with a $\lambda$-dependent basis change.
In terms of Schr\"odinger operators these situations occur on band-edges, cf. Appendix~\ref{sub-Jordan}.

\end{enumerate}
\end{rem}

The proof is structured in the following way. Section~\ref{app} states an abstract theorem for convergence of Markov processes to SDE limits that we will use.
In Section~\ref{sec-evol} we will develop the evolution equations for the process $(X_{\lambda,n}, Z_{\lambda,n})$, together with
some crucial estimates. In Section~\ref{sec-lim} we will then obtain the limiting stochastic differential equations
as in Theorem~\ref{theo-SDE}. The reader interested in the proofs can continue with Section~\ref{app}.

Applications to random Schr\"odinger operators are given in the following Subsection and in Section~\ref{sec-RSO} which also contains the proofs.


\subsection{The GOE limit for random Schr\"odinger operators \label{sub-RS}}

\newcommand{\am}{\mathbb Z_{n,d} }
Let $\am$ be the adjacency matrix of the $n\times d$ grid, and let $V$ be a diagonal matrix with i.i.d. random entries of the same dimension.
A fundamental question in the theory of random Schr\"odinger operators is how the eigenvalues of
\begin{equation}\label{eq:rsch}
\am + \lambda V
\end{equation}
are distributed. Predictions from the physics literature suggest that in certain scaling regimes that correspond to the delocalized regime, random matrix behavior should appear. More precisely,
the random set of eigenvalues, in a window centered at some energy $E$ and scaled properly, should converge to the Sine$_1$ point process. The latter process is the large-$n$ limit of the random set of eigenvalues of the $n\times n$ Gaussian orthogonal ensemble near 0.

A version of such predictions were proved rigorously \cite{VV1} for subsequences of $n_i\gg d_i\to\infty$, $\lambda_i^2 n_i\to 0$  but only near energies $E_i$ tending to zero. In a modified model where  the edges in the $d$ direction get weight $r<1$ the proof of \cite{VV1} works for almost all energies in the range $(-2+2r,2-2r)$.
Proving such claims for almost all energies of the original model \eqref{eq:rsch} presented a challenge, the main motivation for the present paper. 
For better comparison with \cite{VV1} let us re-introduce the weight $r$.
It is natural to think of operators like \eqref{eq:rsch} as acting on a sequence $\psi=(\psi_1,\ldots,\psi_n)$ of $d$-vectors. 
So given the weight $r$, let us define the $nd\times nd$ matrix $H_{\lambda,n,d}$ by
\begin{equation}\label{eq-H-lnd}
 (H_{\lambda,n,d} \psi)_k = \psi_{k+1} + \psi_{k-1} + (r\ZZ_d + \lambda V_k )\psi_k
\end{equation}
with the notational convention that $\psi_0=\psi_{n+1}=0$.
Here, $\ZZ_d$ is the adjacency matrix of the connected graph of a path with $d$ vertices
and the $V_k$ are i.i.d. real diagonal matrices, i.e.,
\begin{equation}\label{eq-GOE-cond1}
 \mathbb Z_d=\pmat{\nul  & \one &  & \\
 \one & \ddots & \ddots & \\
 & \ddots & \ddots & \one \\
  & & \one & \nul}\;,\qquad
 V_1=\pmat{v_1 & & \\
 & \ddots & \\
  & & v_d}
\end{equation}
with
\begin{equation} \label{eq-GOE-cond2}
\EE(v_j)=0,\quad\EE(|v_j|^{6+\varepsilon})<\infty,\qquad\EE(v_i v_j)=\delta_{ij}.
\end{equation}
Then we obtain the following:
\begin{theo}\label{theo-GOE}
 For any fixed $r>0$ and almost every energy $E\in (-2-2r,2+2r)$ there exist sequences  $n_k \gg d_k \to \infty$,\, $\sigma_k^2:=\lambda_k^2 n_k \to 0,\,$ 
 and normalizing factors $\nu_k\sim d_k n_k / \sigma_k$ such that
 the process of eigenvalues of
 $$\nu_k \left( H_{\lambda_k,n_k,d_k} - E\right)$$
 converges to the ${\rm Sine}_1$ process.
 In particular the level statistics corresponds to GOE statistics in this limit.
\end{theo}

Theorem \ref{theo-GOE} resolves Problem~3 posed in {\rm \cite{VV1}}. There one has $r<1$ and $E\in (-2+2r, 2-2r)$  or $r=1$ and a sequence of energies converging to $0$. 
(Note that this interval is smaller than the one in Theorem~\ref{theo-GOE} and in fact empty for $r\geq 1$). Theorem~\ref{theo-GOE} applies to the exact Anderson model $r=1$ with any fixed energy in the interior $(-4,4)$ of the spectrum of the discrete two-dimensional Laplacian. It also applies in the case $r>1$. This is because hyperbolic channels can now be handled for the SDE limit. 
The exact definition of 'elliptic' and 'hyperbolic' channels will be given in Section~\ref{sec-RSO}.
Overcoming this difficulty was the main motivation for this work.

Essentially, only the elliptic channels play a role in the eigenvalue process limit. It is thus important to have a sequence with a growing number of elliptic channels going to infinity.
Indeed, one can obtain GOE statistics even for a sequence of energies $E_k$ approaching the edge of the spectrum $|E|=2+2r$. For this, one needs that the sequence $d_k$ grows fast enough, such that the number of elliptic channels at energy $E_k$ grows.

Further details will be given in Section~\ref{sec-RSO} where we will also consider an SDE description for the eigenvalue processes of operators on strips with a fixed width in the critical scaling. 
The operators considered are slightly more general than \eqref{eq-H-lnd}.



\section{A limit theorem for Markov processes \label{app}}

The key idea is to use a variation of Proposition~23 in \cite{VV2} to obtain the convergence to the limiting process.

\begin{prop}\label{p_turboEK}
Fix $T>0$, and for each $m\ge 1$ consider a Markov chain
$$
(X^m_n\in \RR^d,\, n =1\ldots  \lfloor mT
\rfloor).
$$
as well as a sequence of ``good'' subsets $G_m$ of $\RR^d$.
Let $Y^m_n(x)$ be distributed as the increment
$X^m_{\ell+1}-x$ given $X^m_n=x\in G_m$. We define
$$
b^m(t,x)= m \EE[ Y_{\lfloor mt\rfloor}^m(x)],\qquad
a^m(t,x)=m \EE[ Y_{\lfloor mt\rfloor}^m(x) Y_{\lfloor
mt\rfloor}^m(x)^\textup{T}].
$$
Let $d'\le d$, and let $\tilde x$ denote the first $d'$ coordinates
of $x$. These are the coordinates that will be relevant in the limit.
Also let $\tilde b^m$ denote the first $d'$ coordinates of $b^m$ and $\tilde a^m$ be the
upper left $d'\times d'$ sub-matrix of $a^m$.

Furthermore, let $f$ be a function $f:\ZZ_+\to\ZZ_+$ with $f(m)=o(m)$, i.e $\lim_{m\to\infty} f(m) / m = 0$
Suppose that as $m\to \infty $ for $x,y\in G_n$ we have
\begin{align}
& |\tilde a^m(t,x)-\tilde a^m(t,y)|+|\tilde b^m(t,x)-\tilde b^m(t,y)| \le
c|\tilde x-\tilde y|+o(1)\label{e lip}\\
 &\sup_{x\in G_m,n}  \EE[|\tilde Y^m_n(x)|^3] \le cm^{\frac{-3}{2}}\quad\text{for all}\quad
 n\geq f(m) \label{e 3m},
\end{align}
and that there are functions $a,b$ from $\RR\times[0,T]$ to
$\RR^{d'^2}, \RR^{d'}$ respectively with bounded first and second
derivatives so that uniformly for $x\in G_m$,
\begin{eqnarray}
\sup_{x\in G_m, t} \Big|\int_0^t \tilde a^m(s,x)\,ds-\int_0^t a(s,\tilde x)\,ds
\Big| &\to& 0 \\
\sup_{x\in G_m,t} \Big|\int_0^t \tilde b^m(s,x)\,ds-\int_0^t b(s,\tilde x)\,ds
\Big|&\to& 0. \label{e fia}\,.
\end{eqnarray}
Suppose further that
$$
\tilde X_{f(m)}^m\Longrightarrow X_0.
$$
and that $\PP(X^m_n \in G_m \mbox{ for all }n\geq f(m))\to 1$.
Then $(\tilde X^m_{\lfloor m t\rfloor}, 0 < t \le T)$ converges in law
to the unique solution of the SDE
$$
dX = b \,dt + a\, dB, \qquad X(0)=X_0.$$
\end{prop}
\begin{proof}
This is essentially Proposition 23 in \cite{VV2}. The first difference is
that the coordinates $d'+1,\ldots, d$ of the $X^m$ do not appear in
the limiting process. A careful examination of the proof of that Proposition shows
that it was not necessary to assume that all coordinates appear in the limit, as long as
the auxiliary coordinates do not influence the variance and drift asymptotics.

The second difference is the introduction of the ``good'' set $G_m$, possibly a proper subset of $\RR^d$. Since we assume
that the processes $X^m$ stay in $G_m$ with probability tending to one, we can apply the Proposition 23 of \cite{VV1}
to $X^m$ stopped when it leaves this set. Then, the probability that the stopped process is different from the original
tends to zero, completing the proof.

The third difference is the weak convergence of $\tilde X_{f(m)}^m$ instead of $\tilde X^n_0$
and that we have the bound in \eqref{e 3m} only for $m\geq f(m)$.
Note that for the Markov family $\hat X^{m}_l = \hat X^{m}_{\max(l,f(m))}$
all the same conditions apply with $f(m)=0$ and the initial conditions converge weakly.
Moreover, for any fixed $t>0$ and $m$ large enough one has
$\hat X^m_{\lfloor mt \rfloor}=X^m_{\lfloor mt \rfloor}$.
\end{proof}

We will use this proposition with $m=\lambda^{-2}$ or $\lambda=1/\sqrt{m}$. $X^m_n$ will correspond to the pair $(X_{1/\sqrt{m}\,,\,n}\,,\,Z_{1/\sqrt{m},n})$
whereas $\tilde X^m_n$ will only be the part of $X_{1/\sqrt{m}\,,\,n}$ giving the SDE limit in Theorem~\ref{theo-SDE}.
Moreover, many of the estimates will only work with high probability which will be treated by introducing stopping times that will not matter in the limit.


\section{Evolution equation and estimates \label{sec-evol}}

In this section we show the basic estimates needed to establish the conditions of the proposition above.
Recall that $\Tt_{\lambda,n}=\Tt_0+\lambda \Vv_{\lambda,n} + \lambda^2 \Ww_\lambda$ where the disordered part
satisfies the assumptions \eqref{eq-cond-Vv-m} and \eqref{eq-cond-Vv-lim}.
For convenience we define $\Yy_{\lambda,n} = \Vv_{\lambda,n}+\lambda \Ww_\lambda$, such that
\begin{equation*}
 \Tt_{\lambda,n} = \Tt_0+\lambda \Vv_{\lambda,n} + \lambda^2 \Ww_\lambda=\Tt_0+\lambda\Yy_{\lambda,n}\;.
\end{equation*}
Then the assumptions imply that for small $\lambda$
\begin{equation}\label{eq-cond-Yy}
\EE(\|\Yy_{\lambda,n}\|^{6+\epsilon})=\Oo(1)\;,\qquad
\EE(\Yy_{\lambda,n})=\lambda \Yy+o(1)\;.
\end{equation}

For the proof of Theorem~\ref{theo-SDE}
we will fix some time $T>0$ and obtain the SDE limit up to time $T>0$ which is fixed but arbitrary.
In principle we could work with estimates that are valid with high probability in order to obtain the limit process.
However, for the sake of keeping the arguments and estimates simpler, it will be easier to work with a cut off 
on the randomness and almost sure estimates. The cut off bound will approach infinity for $\lambda$ going to zero in a way that it does not affect the limit.
\begin{prop}
Without loss of generality we can assume
\begin{equation}
\label{eq-cond-Vv}
  \|\Yy_{\lambda,n}\| < K_\Yy \lambda^{s-1}\;\;\text{for some $\tfrac23<s<1$ and $K_\Yy>0$}\,. \\
\end{equation}
\end{prop}
\begin{proof}

Assumption \eqref{eq-cond-Vv-m} and Markov's inequality yield
\begin{equation} \label{eq-prob-est}
 \PP(\|\Yy_{\lambda,n}\| \geq C \|)  \leq \frac{\EE(\|\Yy_{\lambda,n}\|^{6+\epsilon})}{C^{6+\epsilon}}\leq
 \frac{k}{C^{6+\epsilon}}\,
\end{equation}
for some fixed $k$, uniformly for small $\lambda$.

Now let $s$ be such that  $2/(6+\epsilon)\,<\,1-s\,<\,1/3$ and
define the truncated random variable $\widetilde \Yy_{\lambda,n}$ by
\begin{equation*}
 \wt \Yy_{\lambda,n} = \begin{cases}
                        \Yy_{\lambda,n} & \qtx{if} \|\Yy_{\lambda,n}\|<K\lambda^{s-1} \\
                        0 & \qtx{else}
                       \end{cases}
\end{equation*}
By the choice of $s$, $(1-s)>\frac2{6+\epsilon}>\frac1{5+\epsilon}$ and we obtain
\begin{align*}
 \EE(\|\Yy_{\lambda,n}-\wt\Yy_{\lambda,n} \|)\,&=\,
 \int_{\|\Yy_{\lambda,n}>K\lambda^{s-1}\|} \|\Yy_{\lambda,n} \|\,d\PP
 \,\leq\, \int_{K\lambda^{s-1}}^\infty C \,\cdot\,(6+\epsilon)\frac{k}{C^{7+\epsilon}}\,dC \notag\\
 &\leq\, \frac{(6+\epsilon) k}{(5+\epsilon)\,K^{5+\epsilon}}\,\lambda^{(5+\epsilon)(1-s)}\,=\,o(\lambda)\;
\end{align*}
and similarly
\begin{equation*}
 \EE(\|\Yy_{\lambda,n}-\wt\Yy_{\lambda,n} \|^2)\,\leq\,
 \frac{(6+\epsilon)\,k}{(4+\epsilon)\,K^{4+\epsilon}}\,\lambda^{(4+\epsilon)(1-s)}\,=\,o(1)
\end{equation*}
for $\lambda\to 0$. Thus, using $\wt \Yy_{\lambda,n}$ instead of $\Yy_{\lambda,n}$ in \eqref{eq-def-V}, \eqref{eq-def-g} and
\eqref{eq-def-hg} does not change the quantities $V$, $g(M)$ and $\hat g(M)$.
Hence, the SDE limits mentioned in Theorem~\ref{theo-SDE} for $\wt \Yy_{\lambda,n}$ and $\Yy_{\lambda,n}$ are the same.

Let us assume that Theorem~\ref{theo-SDE} is correct for $\wt \Yy_{\lambda,n}$ and obtain its validity for using
$\Yy_{\lambda,n}$ by showing that we obtain the same limit SDE.
From \eqref{eq-prob-est}
\begin{equation*}
\PP(\Yy_{\lambda,n} \neq \wt \Yy_{\lambda,n})\,\leq \frac{k}{K^{6+\epsilon} \lambda^{(6+\epsilon)(s-1)}}\,
\,=\,c\,\lambda^{(6+\epsilon)(1-s)}=c \lambda^{2+\delta}
\end{equation*}
where the last equations define $c>0$ and $\delta>0$.
Hence,
\begin{equation*}
\PP\Big(\|\Yy_{\lambda,n}\| > K \lambda^{s-1}\;\;\;
\text{for some $n=1,2,\ldots,\lfloor\lambda^{-2}T\rfloor$}\Big)\,\leq\, Tc\lambda^\delta
\end{equation*}
which approaches zero for $\lambda\to 0$.
Therefore, introducing a stopping time $T_\lambda:=\min\{n\,:\, \Yy_{\lambda,n}\neq\wt\Yy_{\lambda,n}\}$
and considering the stopped process $X_{\lambda,n \wedge T_\lambda}$ one obtains the same distributional limit process up to time $T$ (in fact for arbitrary $T$).
But for the stopped processes there is no difference when replacing $\Yy_{\lambda,n}$ by $\wt \Yy_{\lambda,n}$.
\end{proof}

Thus we may assume equation \eqref{eq-cond-Vv} without loss of generality and we will do so from now on.
 Moreover, as the spectral radius of $\Gamma_0$ and $\Gamma_2$ are smaller than 1,
 using a basis change, we may assume:\footnote{even if $\Gamma_0$ or $\Gamma_2$ are not diagonalizable, one can make the norm smaller than one as one can make the off diagonal terms of Jordan blocks arbitrarily small.}
\begin{equation}\label{eq-def-gamma}
 \|\Gamma_0\|\leq e^{-\gamma} \;,\quad
 \|\Gamma_2\|\leq e^{-\gamma}\;,\quad \text{where $\gamma>0$}.
\end{equation}

\vspace{.2cm}

Before obtaining the evolution equations, we will first establish that the pair
$(X_{\lambda,n},Z_{\lambda,n})$ is a Markov process.
Let us define the following subgroup of $\GL(d,\CC)$.
\begin{equation*}
\GG=\left\{ \begin{pmatrix}
             \one & 0 \\ C & D
            \end{pmatrix} \,\in\,{\rm Mat}(d,\CC) \;:\quad \text{where}\quad
            D \in \GL(d_2,\CC)\;
\right\}\;.
\end{equation*}
Now let $\Xx_1$ and $\Xx_2$ be equivalent, $\Xx_1\sim \Xx_2$, if
$\Xx_1=\Xx_2\Qq$ for $\Qq\in\GG$.
As different representatives differ by multiplication from the right, multiplication from the left defines an action on the
equivalence classes.
Therefore, the evolution of the equivalence classes
$[\Xx_{\lambda,n}]_\sim$ is a Markov process.
As
\begin{equation} \label{eq-std-form}
\begin{pmatrix}
 A & B \\ C & D
\end{pmatrix}
\begin{pmatrix}
 \one & \nul \\ -D^{-1} C & D^{-1}
  \end{pmatrix}
\,=\,\begin{pmatrix}
  A-BD^{-1}C & BD^{-1} \\ \nul & \one
 \end{pmatrix}
\end{equation}
we see that the equivalence class $[\Xx_{\lambda,n}]_\sim$ is determined by
the pair $(X_{\lambda,n}, Z_{\lambda,n})$.

\vspace{.2cm}

Let us further introduce the following commuting matrices of size $d_0+d_1$,
\begin{equation} \label{eq-def-R-S}
R=\pmat{\one & \nul \\ \nul & U},\quad
S=\pmat{\Gamma_0 & \nul \\ \nul & \one};\quad R,S \in {\rm Mat}(d_0+d_1,\CC)\;,
\end{equation}
Note that $R$ is unitary and that
\begin{equation} \label{eq-def-Rr}
\Rr = \pmat{R & \nul \\ \nul & \one} \in {\rm U}(d)\;.
\end{equation}
for $\Rr$ as defined in \eqref{eq-def-Xx}.
As $\Gamma_0^{n}$ is exponentially decaying,
we refer to the $d_0$ dimensional subspace corresponding  to this matrix block
as the decaying directions of $\Tt_0^n$. Similarly, the $d_2$ dimensional subspace corresponding
to the entry $\Gamma_2^{-n}$ are referred to as growing directions.

\newcommand{\va}{V_{\lambda,n}^A}
\newcommand{\vb}{V_{\lambda,n}^B}
\newcommand{\vc}{V_{\lambda,n}^C}
\newcommand{\vd}{V_{\lambda,n}^D}
\newcommand{\vxs}{V}
\newcommand{\vx}{\vxs_{\lambda,n}^X}

\newcommand{\hva}{R^{-n} V_{\lambda,n}^A R^{n-1}}
\newcommand{\hvb}{R^{-n} V_{\lambda,n}^B}
\newcommand{\hvc}{V_{\lambda,n}^C R^{n-1}}
\newcommand{\hvd}{V_{\lambda,n}^D}

\newcommand{\ta}{T_{\lambda,n}^A}
\newcommand{\tb}{T_{\lambda,n}^B}
\newcommand{\tc}{T_{\lambda,n}^C}
\newcommand{\td}{T_{\lambda,n}^D}

The evolution of $\Xx_{\lambda,n}$ is given by
\begin{equation} \label{eq-evol-Xx}
 \Xx_{\lambda,n}= \Rr^{-n} \Tt_{\lambda,n} \Rr^{n-1} \,\Xx_{\lambda,n-1}\,.
\end{equation}
Therefore, let
\begin{equation*}
 \Rr^{-n} \Tt_{\lambda,n} \Rr^{n-1} = \pmat{\ta & \tb \\ \tc & \td }\,.
\end{equation*}
Here, $A,B,C,D$ are used as indices to indicate that we use the same sub-division of the matrix as we did when defining
$A_{\lambda,n},\,B_{\lambda,n},\,C_{\lambda,n}$ and $D_{\lambda,n}$.

The action on the equivalence class of $\Xx_{\lambda,n-1} \sim \smat{X_{\lambda,n-1} & Z_{\lambda,n-1} \\ \nul& \one}$ gives
\begin{equation*}
 \pmat{\ta & \tb \\ \tc & \td } \pmat{X_{\lambda,n-1} & Z_{\lambda,n-1} \\ \nul& \one}=
 \pmat{\ta X_{\lambda,n-1} & \ta Z_{\lambda,n-1} + \tb \\ \tc X_{\lambda,n-1} & \tc Z_{\lambda,n-1} + \td}\,.
\end{equation*}
Transforming the matrix on the right hand side into the form as in \eqref{eq-std-form} we can read off
the evolution equations
\begin{equation} \label{eq-evol-Z}
 Z_{\lambda,n} = \left(\ta Z_{\lambda,n-1} + \tb \right)\,\left(\tc Z_{\lambda,n-1} + \td \right)^{-1}\;
\end{equation}
and
\begin{equation}
X_{\lambda,n} =  \ta X_{\lambda,n-1} - \left(\ta Z_{\lambda,n-1} + \tb \right) \left(\tc Z_{\lambda,n-1} + \td \right)^{-1}
 \tc X_{\lambda,n-1} \label{eq-exp-X0}\,.
\end{equation}

For more detailed calculations, let
\begin{equation}\label{eq-def-va}
\Yy_{\lambda,n}=\Vv_{\lambda,n}+\lambda\Ww_\lambda=:
\pmat{\va & \vb \\ \vc & \vd  }\;,\;\;\text{then}\quad
\Rr^{-n} \Yy_{\lambda,n} \Rr^{n-1} =\pmat{\hva & \hvb \\ \hvc & \hvd  }\,.
\end{equation}
From \eqref{eq-def-Tt}, \eqref{eq-T_0}, \eqref{eq-def-R-S} and \eqref{eq-def-Rr} one finds
\begin{equation}\label{eq-rel-t-v}
 \ta=S+\lambda\hva\;,\;\; \tb=\lambda\hvb\;,\;\; \tc=\lambda\hvc\;,\;\; \td=\Gamma_2^{-1}+\lambda\hvd\;.
\end{equation}

We will first consider the Markov process $Z_{\lambda,n}$ and denote the starting point by
$Z_0=B_0D_0^{-1}$.

\begin{prop}\label{prop-Z-small}
 For $\lambda$ small enough and some constant $K_Z$ we have the uniform bound
\begin{equation}\label{eq-Z-small}
 \|Z_{\lambda,n}\| \leq K_Z(e^{-\gamma n/2}+\lambda^s)\,=\,\Oo(e^{-\gamma n/2}, \lambda^s)\,
\end{equation}
with $\gamma$ as in \eqref{eq-def-gamma}.
This implies $Z_{\frac1{\sqrt{m}},\lfloor tm\rfloor}\Rightarrow 0$ in law.
\end{prop}

\begin{proof}
Take $\lambda$ small enough, such that $e^\gamma - K_\Yy\lambda^{s} (1+\max(\|Z_0\|,1)) > e^{\gamma/2}$
(with $K_\Yy$ as in \eqref{eq-cond-Vv}) and
\begin{equation}\label{eq-cond-lb}
 \frac{(1+K_\Yy\lambda^s)\max(\|Z_0\|,1)+K_\Yy\lambda^s}{e^{\gamma}-K_\Yy\lambda^s\left (1+\max(\|Z_0\|,1)\right)}
 \,\leq\,e^{-\gamma/2}\,\max(\|Z_0\|,1).
\end{equation}
Then, using \eqref{eq-evol-Z}
we find for $\|Z_{\lambda,n-1}\|\leq \max(\|Z_0\|,1)$
\begin{equation*}
 \|Z_{\lambda,n}\| \leq  \frac{(1+K_\Yy\lambda^s)\|Z_{\lambda,n-1}\|+K_\Yy\lambda^s}
 {e^{\gamma}-K_\Yy\lambda^s\|Z_{\lambda,n-1}\|}\,\leq e^{-\gamma/2} \max(\|Z_0\|,1) \,<\, \max(\|Z_0\|,1).
\end{equation*}
Hence, inductively, $\|Z_{\lambda,n}\|<\max(\|Z_0\|,1)$
for $n\geq 1$. Thus, using this equation and \eqref{eq-cond-lb} again leads to
\begin{equation*}
 \|Z_{\lambda,n}\| \leq e^{-\gamma/2} \|Z_{\lambda,n-1}\|+K_\Yy \lambda^s\,.
\end{equation*}
By induction this yields the bound
\begin{equation*}
 \|Z_{\lambda,n}\| \leq e^{-\gamma n/2}\|Z_{0}\|\,+\, \frac{1-e^{-\gamma n/2}}{1-e^{-\gamma/2}} K_\Yy \lambda^s
\end{equation*}
proving the proposition.
\end{proof}

\begin{rem}
 Note that the estimates show that $\tc Z_{\lambda,n-1}+\td$ is invertible. Using
$D_{\lambda,n} = \tc B_{\lambda,n-1}+\td D_{\lambda,n-1}=(\tc Z_{\lambda,n-1}+\td)D_{\lambda,n-1} $
it follows inductively also that $D_{\lambda,n}$ is invertible.
Hence, $X_{\lambda,n}$ and $Z_{\lambda,n}$ are always well defined for small $\lambda$ under assumption \eqref{eq-cond-Vv}.
Hence, under the assumptions of Theorem~\ref{theo-SDE} they will be well defined up to $n=T\lambda^{-2}$
with probability going to one as $\lambda\to 0$.
\end{rem}

\vspace{.2cm}

Next, let the reminder term
$\wt\Xi_{\lambda,n}$ be given by
\begin{equation}\label{eq-def-tl-Xi}
(\td+\tc Z_{\lambda,n-1})^{-1} %
=\Gamma_2 + \wt\Xi_{\lambda,n} ,
\end{equation}
and define
\begin{equation}\label{eq-def-Xi}
\Xi_{\lambda,n} := -\ta Z_{\lambda,n-1} \Gamma_2 \hvc-(\ta Z_{\lambda,n-1}+\lambda\hvb)\wt\Xi_{\lambda,n} \hvc\,.
\end{equation}
Furthermore let
\begin{equation}\label{eq-def-W_lb}
 \vx:= \va-\lambda \vb \Gamma_2 \vc\,.
\end{equation}
The upper index $X$ should indicate that this is the important combination of the random parts $\Yy_{\lambda,n}$ that will contribute to the SDE limit for the process $X_{\lambda,n}$.
As we will establish, the 'reminder' part expressed in the $\Xi_{\lambda,n}$ terms will be of too low order and not matter in the limit.
By \eqref{eq-exp-X0} and \eqref{eq-rel-t-v} one obtains
\begin{equation} \label{eq-exp-X}
X_{\lambda,n}=S X_{\lambda,n-1} + \lambda R^{-n} \vx R^{n-1} X_{\lambda,n-1}
+ \lambda\Xi_{\lambda,n} X_{\lambda,n-1}\,.
\end{equation}
The following estimates will be needed to obtain the SDE limit.

\begin{lemma}\label{lem-estimates}
Let $\EE_{X,Z}$ denote the conditional expectation given that $X_{\lambda,n-1}=X$
and $Z_{\lambda,n-1}=Z$.
\begin{enumerate}[{\rm (i)}]
\item For small $\lambda$ one has the bounds
\begin{align}
 \label{eq-E-Xi} \EE_{X,Z}(\Xi_{\lambda,n}) &= \EE(\Xi_{\lambda,n}|Z_{\lambda,n-1}=Z) = \Oo(\lambda^{2s-1}\|Z\|,\lambda^{3s-1})\\
 \label{eq-E-Xi2}
 \EE_{X,Z} (\|\Xi_{\lambda,n}\|^2) &= \Oo(\|Z\|^2,\lambda^2,\lambda\|Z\|) \\
 \label{eq-Xi-bound} \Xi_{\lambda,n} &= \Oo(\|Z_{\lambda,n-1}\|\lambda^{s-1},\lambda^{3s-1})\,.
\end{align}
\item  $\vx$ is independent of $Z_{\lambda,n-1}$ and $X_{\lambda,n-1}$ and
there is a matrix $\vxs_0$ and a constant $K$ such that
\begin{align}
 \label{eq-E-W} \EE(\vx) &=\lambda \vxs_0 + o(\lambda) \\
 \label{eq-W-bound} \vx & = \Oo(\lambda^{s-1}) \\
 \label{eq-W3-bound}\EE(\|\vx\|^3) &\leq K =\Oo(1)\;.
\end{align}
\item
\begin{equation}
\label{eq-XX-bound} \EE_{X,Z}\Big(X_{\lambda,n}X^*_{\lambda,n}\Big)
 = SXX^*S+\|X\|^2\cdot \Oo(\lambda^2,\lambda^{2s}\|Z\|).
\end{equation}
Moreover,
there is a function $K(T)$ such that
\begin{equation}
\label{eq-XX-bound1} \EE(\|X_{\lambda,n}\|^2) \leq K(T)\quad
\text{for all $n< T \lambda^{-2}$}.
\end{equation}

\end{enumerate}
\end{lemma}

\begin{proof}
Note that \eqref{eq-cond-Vv} implies the uniform bounds
\begin{equation}\label{eq-bd-va}
V_{\lambda,n}^\#=\Oo(\lambda^{s-1}) \qtx{for} \#\in\{A,B,C,D,X\}
\end{equation}
As $\td=\Gamma_2^{-1}+\lambda\hvd$, $\ta=S+\lambda\hva$ one finds
\begin{equation}\label{eq-xi-est1}
 \wt \Xi_{\lambda,n}=\Oo(\lambda^s)\quad\text{and}\quad
 (\ta Z+\lambda \hvb)\wt\Xi_{\lambda,n}\hvc =\Oo(\lambda^{3s-1},\lambda^{2s-1}\|Z\|)\;.
\end{equation}
Using \eqref{eq-cond-Vv-lim} we see that
$\EE(\ta Z \Gamma_2 \hvc)=\Oo(\lambda\|Z\|)$
which together with \eqref{eq-xi-est1} gives \eqref{eq-E-Xi}.
(Note that $\va,\,\vb$ and $\vc$ are independent of $Z_{\lambda,n-1}$.)
The moment condition \eqref{eq-cond-Vv-m} also yields
$\EE (\|\ta Z \Gamma_2 \hvc\|^2) = \Oo(\|Z\|^2)$.
Combining this with \eqref{eq-xi-est1}, using Cauchy Schwarz in the form
$\EE(\|A+B\|^2) \leq \left(\sqrt{\EE(\|A\|^2)} + \sqrt{\EE(\|B\|^2)}\right)^2$
and using $\Oo(\lambda^{3s-1},\lambda^{2s-1}\|Z\|) \leq \Oo(\lambda,\|Z\|)$ we find for some constant $K$ that
$\EE_{X,Z}(\|\Xi_{\lambda,n}\|^2) \leq K (\|Z\|+\lambda)^2$ giving \eqref{eq-E-Xi2}.
Finally, \eqref{eq-bd-va} yields $\|\ta Z \Gamma_2 \hvc\|=\Oo(\|Z\|\lambda^{s-1})$ which
combined with \eqref{eq-xi-est1} gives \eqref{eq-Xi-bound}.

To get (ii) note that equation \eqref{eq-E-W} follows from
\eqref{eq-cond-Vv-lim}, \eqref{eq-bd-va} yields \eqref{eq-W-bound} and
the moment condition \eqref{eq-cond-Vv-m} implies \eqref{eq-W3-bound}.

For part (iii) note that by \eqref{eq-exp-X} one has
\begin{align}
& \EE_{X,Z} (X_{\lambda,n} X_{\lambda,n}^*) = SX X^* S + \lambda R^{-n} \EE(\vx) R^{n-1} X X^*S
+\lambda \left[R^{-n} \EE(\vx) R^{n-1} X X^*S \right]^* \notag \\
& \qquad + \lambda \EE_{X,Z}(\Xi_{\lambda,n}) X X^* S + \lambda \left[ \EE(\Xi_{\lambda,n}) X X^* S\right]^* \,+\,
\Oo\left(\lambda^2\|X\|^2 \EE_{X,Z}\big((\|\vx \|+\|\Xi_{\lambda,n}\|)^2\big) \right) \notag\,.
\end{align}
Using \eqref{eq-E-Xi}, \eqref{eq-E-W} and $\EE_{X,Z}((\|\vx\|+\|\Xi_{\lambda,n}\|)^2)=\Oo(1) $
one finally obtains equation \eqref{eq-XX-bound}. The latter estimate follows from
\eqref{eq-Z-small},\,\eqref{eq-Xi-bound}, \eqref{eq-W3-bound} and Cauchy-Schwarz.\\
For \eqref{eq-XX-bound1} note that the Hilbert-Schmidt norm is given by
$\|X\|^2_{HS}=\Tr(XX^*)$. Then \eqref{eq-Z-small} and \eqref{eq-XX-bound} imply
that for some constant $K$ one finds
$$
 \EE\left(\|X_{\lambda,n}\|_{HS}^2\right)\leq \begin{cases}
\EE\left(\|X_{\lambda,n-1}\|^2_{HS}\right) (1+K \lambda^{2s}) &
\quad \text{for $n \leq s/\gamma \ln(\lambda^{-2})$} \\
\EE\left(\|X_{\lambda,n-1}\|^2_{HS}\right) (1+K \lambda^{2}) &
\quad \text{for $n > s/\gamma \ln(\lambda^{-2})$.}
\end{cases}
$$
By induction, for small $\lambda$ and $n<T\lambda^{-2}$,
$$
\EE(\left(\|X_{\lambda,n}\|_{HS}^2\right)\leq
(1+K \lambda^{2s})^{s/\gamma \ln(\lambda^{-2})}(1+K \lambda^{2})^{T\lambda^{-2}}\|X_0 \|^2
<e^{K+TK}\,\|X_0\|^2.
$$
As all norms are equivalent, this finishes the proof.
\end{proof}



\section{Proof of Theorem~\ref{theo-SDE}, the limit of $X_{\lambda,n}$ \label{sec-lim}}

We need to split up the $(d_0+d_1)\times(d_0+d_1)$ matrix $X_{\lambda,n}$ into the corresponding blocks.
Therefore, let
\begin{equation*}
P_0=\pmat{\one & \nul }\,\in\, {\rm Mat}(d_0 \times (d_0+d_1))\,,\quad
P_1=\pmat{\nul & \one }\,\in\, {\rm Mat}(d_1 \times (d_0+d_1))
\end{equation*}
Then, using \eqref{eq-def-R-S} one finds
\begin{equation}\label{eq-P-ids}
P_0 S = \Gamma_0 P_0,\quad P_1 S  = P_1,\quad
P_0R^n= P_0,\quad P_1 R^n = U^n P_1\,.
\end{equation}
Moreover,
for any $(d_0+d_1)\times (d_0+d_1)$ matrix $M$ one has
\begin{equation}\label{eq-P-ids2}
M=\pmat{P_0 M \\ P_1 M}=\pmat{MP_0^* & MP_1^*}\,.
\end{equation}
\begin{prop} \label{prop-P0X}
There is a function $K(T)$ such that for all $n<T\lambda^{-2}$
one has
$$
 \EE(\|P_0X_{\lambda,n}\|_{HS}^2) \,\leq\, e^{-2\gamma n} \|P_0X_0\|_{HS}^2+ K(T)\lambda^{2s}
$$
In particular, for any function $f(n)\in\NN$ with $\lim_{n\to \infty} f(n)=\infty$ one has
$$
P_0X_{\frac{1}{\sqrt{n}},f(n)}\Longrightarrow
\pmat{\nul&\nul}\,.
$$
\end{prop}

\begin{proof}
Multiplying \eqref{eq-XX-bound} by $P_0$ from the left and $P_0^*$ from the right,
taking expectations and using the bound \eqref{eq-XX-bound1} gives
$$
\EE(P_0 X_{\lambda,n} X_{\lambda,n}^*P_0^* )=
\Gamma_0 \EE(P_0X_{\lambda,n-1} X_{\lambda,n-1}^*P_0^*)\Gamma_0^*+
\Oo(\lambda^{2s})
$$
which leads to
$$
\EE(\|P_0X_{\lambda,n}\|_{HS}^2) \leq e^{-2\gamma} \EE(\|P_0 X_{\lambda,n-1}\|_{HS}^2)+
\Oo(\lambda^{2s})\,.
$$
where the bound for the error term is uniform in $n$ for $n<\lambda^{-2}T$.
Induction yields the stated result.
\end{proof}

Finally, let us consider the part with an interesting limit.
Multiplying \eqref{eq-exp-X} by $P_1$ from the left and and using \eqref{eq-P-ids}, \eqref{eq-P-ids2}
one finds
\begin{align}
 P_1 X_{\lambda,n} &= P_1 X_{\lambda,n-1}+\lambda {U}^{-n} P_1 \vx
 ( P_1^* U^{n-1} P_1 X_{\lambda,n-1}+P_0^*P_0 X_{\lambda,n-1}) \notag \\
 &\quad + \lambda P_1 \Xi_{\lambda,n} (P_1^* P_1 X_{\lambda,n-1}+P_0^* P_0 X_{\lambda,n-1})
 \label{eq-evol-P0X}
\end{align}
We immediately obtain the following estimate.

\begin{prop} \label{prop-P1X}
For $n<\lambda^{-2} T$ one has uniformly
$$
\EE(\|P_1X_{\lambda,n}-P_1X_0\|^2) \leq \Oo(n\lambda^{2s})\,.
$$
This implies for any function $f(n)\in\NN$ with
$\lim_{n\to \infty} f(n) n^{-s} = 0$ that
$$
P_1X_{\frac{1}{\sqrt{n}},f(n) }\,\Longrightarrow\,P_1X_0
$$
in law for $n\to \infty$.
\end{prop}

\begin{proof}
 Using the estimates of Lemma~\ref{lem-estimates} one finds similarly to \eqref{eq-XX-bound} that
\begin{align*}
\EE_{X,Z}((P_1X_{\lambda,n}-P_1X_0) (P_1X_{\lambda,n}-P_1X_0)^*) =
P_1XX^*P_1^* + \|X\|^2 \Oo(\lambda^2,\lambda^{2s} \|Z\|)\,.
\end{align*}
Using \eqref{eq-XX-bound1} and $\|Z_{\lambda,n-1}\|\leq\Oo(1)$ from \eqref{eq-Z-small} we find therefore
 that uniformly for $n<\lambda^{-2} T$
\begin{align*}
&\EE((P_1X_{\lambda,n}-P_1X_0) (P_1X_{\lambda,n}-P_1X_0)^*) = \\
&\qquad \qquad \EE((P_1X_{\lambda,n-1}-P_1X_0)(P_1X_{\lambda,n-1}-P_1X_0)^*) + \Oo(\lambda^{2s})
\end{align*}
Taking traces (Hilbert-Schmidt norm) and induction yield the result.
\end{proof}

In order to use Proposition~\ref{p_turboEK} we need to consider stopped processes.
So for any $\lambda$, let $T_K$ be the stopping time when $\|P_1X_{\lambda,n}\|$ is bigger then $K$.
We define the stopped process by
$$
P_1 X^K_{\lambda,n}:=P_1X_{\lambda,T_K\wedge n}\;,\quad
P_0 X^K_{\lambda,n}:= P_0 X_{\lambda,n} \cdot 1_{n\leq T_K}\;,\quad
Z^K_{\lambda,n} := Z_{\lambda,n} \cdot 1_{n\leq T_K}\;
$$
where
$$
1_{n\leq T_K}= 1 \quad \text{for $n\leq T_K$} \quad\text{and}\quad
1_{n\leq T_K}= 0 \quad \text{for $n>T_K$}\;.
$$
As long as $n\leq T_K$, \eqref{eq-exp-X} and Lemma~\ref{lem-estimates}
give $\|P_0 X_{\lambda,n}\|\leq (e^{-\gamma}+\Oo(\lambda^s)) \|P_0 X_{\lambda,n-1} \|+\Oo(\lambda^s)$.
An induction similar as in Proposition~\ref{prop-Z-small} yields for any finite $K$
\begin{equation}\label{eq-est-P0}
\|P_0 X^K_{\lambda,n}\|<K_P(e^{-\gamma/2 n}+\lambda^s)\;,\quad\text{for some constant $K_P = K_P(K)$.}
\end{equation}

For the limit, we will scale $\lambda=1/\sqrt{m}$ and $n=\lfloor tm \rfloor$.
First, define the good set
\begin{equation*}
G_m=G_m(K):=\{(X,Z)\,:\,\|Z\|<2K_Z m^{-s/2}, \|P_0X\|< 2K_P m^{-s/2}\}\,,
\end{equation*}
then by the estimates \eqref{eq-Z-small} and \eqref{eq-est-P0} one has
\begin{equation}\label{eq-in-Gm}
(X^K_{1/\sqrt{m}\,,n},Z^K_{1/\sqrt{m}\,,n}) \in G_m \quad \mbox{for} \quad n>s\,\ln(m)\,/\,\gamma \;.
\end{equation}

For the variances in the SDE limit we need to recognize the connection to the matrix $V$ and the functions
$g(M),\,\hat g(M)$ as defined in
\eqref{eq-def-V}, \eqref{eq-def-g} and \eqref{eq-def-hg}.
Using the notations as introduced in \eqref{eq-def-Vv} combined with \eqref{eq-def-va} and \eqref{eq-def-W_lb} one obtains
\begin{equation}\label{eq-exp-PWP}
 P_1 \vx P_1^* = (\vl11+\lambda\wl11)-\lambda(\vl12+\lambda\wl12)\Gamma_2(\vl21+\lambda\wl21)\;.
\end{equation}
Therefore, using the functions $h,\,\wh{h}$ as defined in \eqref{eq:h} and $W$ as defined in \eqref{eq-def-W}
one finds
\begin{align}
 & \EE(P_1\vx P_1^*)=\lambda P_1 \vxs_0 P_1^*+o(\lambda)\,=\,\lambda W + o(\lambda)\\
 & \EE((P_1 \vx P_1^*)^\top M P_1 \vx P_1^*) = h(M)+o(1) \;\label{eq-h-1} \\
 & \EE((P_1 \vx P_1^*)^* M P_1 \vx P_1^*) = \wh h(M)+o(1)\;\label{eq-h-2}.
\end{align}
Here the error terms $o(\lambda)$ and $o(1)$ are uniform in the limit $\lambda\to 0$.

Next we have to consider the conditional distribution of the differences $Y_{\lambda,n}=Y_{\lambda,n}(X,Z)$
given that $X_{\lambda,n-1}=X, Z_{\lambda,n-1}=Z$, i.e. for Borel sets of matrices $\Aa$,
$$
\PP(Y_{\lambda,n}(X,Z)\in \Aa):=\PP\big( P_1 X_{\lambda,n}-P_1 X \in \Aa\,\big|\,X_{\lambda,n-1}=X, Z_{\lambda,n-1}=Z\big)\;,
$$
Using \eqref{eq-evol-P0X} one has
\begin{equation} \label{eq-Y}
Y_{\lambda,n} =
\lambda {U}^{-n} P_1 \vx
 ( P_1^* U^{n-1} P_1 X +P_0^*P_0 X)
 + \lambda P_1 (\Xi_{\lambda,n}|Z_{\lambda,n-1}=Z)\,X
\end{equation}
where $(\Xi_{\lambda,n}|Z_{\lambda,n-1}=Z)$ is a random matrix variable distributed as $\Xi_{\lambda,n}$ conditioned to
$Z_{\lambda,n-1}=Z$, this simply means that in \eqref{eq-def-tl-Xi} and \eqref{eq-def-Xi} one replaces $Z_{\lambda,n-1}$ by $Z$.

\begin{prop}\label{prop-Y}
Assume $(X,Z)\in G_m$, $m=\lambda^{-2}$, thus $P_0X=\Oo(\lambda^s)$ and
$Z=\Oo(\lambda^s)$.
Then one finds for $Y_{\lambda,n} = Y_{\lambda,n}(X,Z)$ the uniform estimates (uniform in $X,Z,n$ )
\begin{align}
\label{eq-E-Y} \EE(Y_{\lambda,n})  &= \lambda^2 U^{-n} W U^{n-1} P_1 X + o(\lambda^{2})\\
\label{eq-var-Y} \EE(Y^\top_{\lambda,n} M Y_{\lambda,n}) &=
\lambda^2 (P_1X)^\top {U^\top}^{n-1} \,h\big({\bar U}^{n} M {U}^{-n}\big)\, U^{n-1} P_1X+ o(\lambda^{2})\\
\label{eq-var-Y2} \EE(Y^*_{\lambda,n} M Y_{\lambda,n}) &=
\lambda^2 (P_1X)^* {U^*}^{n-1} \,\wh h\big( U^{n} M {U}^{-n}\big)\, U^{n-1} P_1X+ o(\lambda^{2})\\
\label{eq-Y3} \EE(\|Y_{\lambda,n}\|^3) &\leq \lambda^3 K_Y\|X\|^3 \quad\text{for some uniform $K_Y>0$}.
\end{align}
Note that any covariance of real and imaginary entries of $Y_{\lambda,n}$ can be obtained by
varying $M$ in \eqref{eq-var-Y} and \eqref{eq-var-Y2}.
Moreover , one obtains uniformly for $0<t<T$
\begin{align}
\lim_{m\to\infty} \int_0^t m \EE\left(Y_{\frac1{\sqrt{m}},\lfloor sm \rfloor}\right) \,ds
&=  t V P_1X \label{eq-av-Y} \\
\lim_{m\to\infty} \int_0^t m \EE\left(Y^\top_{\frac1{\sqrt{m}},\lfloor sm \rfloor} M
Y_{\frac1{\sqrt{m}},\lfloor sm \rfloor}\right) \,ds
&=  t (P_1X)^\top g(M) P_1 X \label{eq-av-var} \\
\lim_{m\to\infty} \int_0^t m \EE\left(Y^*_{\frac1{\sqrt{m}},\lfloor sm \rfloor} M
Y_{\frac1{\sqrt{m}},\lfloor sm \rfloor}\right) \,ds
&=  t (P_1 X)^* \wh g(M) P_1X \label{eq-av-var2}
\end{align}
where $V$, $g$ and $\wh g$ are as in \eqref{eq-def-V}, \eqref{eq-def-g} and \eqref{eq-def-hg}
\end{prop}

\begin{proof}
Given $X_{\lambda,n-1}=X,\,P_0X=\Oo(\lambda^s),\,Z_{\lambda,n-1}=Z=\Oo(\lambda^s)$
and using the estimates \eqref{eq-E-Xi} and \eqref{eq-E-W}
equation \eqref{eq-Y} yields
$$
\EE(Y_{\lambda,n})=\lambda \EE(U^{-n} P_1 \vx P_1^* U^* U^{n} P_1 X) +
\Oo(\lambda^{3s})=\lambda^2 U^{-n} W U^{n-1} P_1 X\,+ o(\lambda^2)\,
$$
which implies \eqref{eq-E-Y}.
Using \eqref{eq-W-bound}, \eqref{eq-Xi-bound} and $Z_{\lambda,n-1}=Z=\Oo(\lambda^s)$
we get
\begin{equation} \label{eq-exp-Y}
Y_{\lambda,n} = \lambda U^{-n} P_1 \vx P_1^* U^{n-1} P_1 X +\Oo(\lambda^{2s})\,.
\end{equation}
Together with \eqref{eq-h-1} and \eqref{eq-h-2} this proves \eqref{eq-var-Y} and \eqref{eq-var-Y2}.
Finally, \eqref{eq-Y3} follows from \eqref{eq-exp-Y} and \eqref{eq-cond-Vv-m}.

Letting $u=U^{-n}={U^n}^*$ we have the terms
$u W U^*u^*$, $\ov{u} \ov{U} h(u^\top M u) U^* u^* $ and
$u U \wh{h}(u^* M u) U^* u^* $ appearing in \eqref{eq-E-Y}, \eqref{eq-var-Y} and \eqref{eq-var-Y2}, respectively.
On the abelean compact group $\langle U\rangle$ generated by the unitary $U$, the functions
\begin{align*}
& u \mapsto u W u^*\,,
\quad u \mapsto \ov{u}\, \ov{U}\, h(u^\top M u)\, U^*\, u^*\;,\quad
u\mapsto u\, U\, \wh{h}(u^* M u)\, U^*\, u^*
\end{align*}
are polynomials of the eigenvalues of $u$ as $h,\,\wh{h}$ are linear and
all $u\in\langle U \rangle$ are simultaneously diagonalizable.
For any such polynomial $p(u)$ one finds
\begin{equation}\label{eq-U-average}
\lim_{m\to\infty}\int_0^t p(U^{-\lfloor ms \rfloor})\,ds  =
\lim_{m\to\infty} \frac{t}m \sum_{k=1}^m p(U^{-k})\,=\,t\,\int_{\langle U \rangle} p(u)\,du
\end{equation}
uniformly for $t<T$, where $du$ denotes the Haar measure on $\langle U \rangle$.
Applied to \eqref{eq-E-Y}, \eqref{eq-var-Y}, \eqref{eq-var-Y2}
this yields \eqref{eq-av-Y}, \eqref{eq-av-var} and \eqref{eq-av-var2}.
\end{proof}

Let $f(n)\in\NN$ with $\lim_{n\to\infty} f(n)=\infty$ and $\lim_{n\to\infty} f(n) n^{-s}=0$,
then by Proposition~\ref{prop-P0X} and \ref{prop-P1X} we find for large enough $K$ that
$X^K_{1/\sqrt{n},f(n)} \Longrightarrow \smat{0 & 0 \\ 0 & \one_{d_1}} X_0$ where we used a
subdivision in blocks of sizes $d_0$ and $d_1$.
For sake of concreteness let us set $f(n)=\lfloor n^\alpha \rfloor$ with some $0<\alpha< s$.
From \eqref{eq-in-Gm} we find for $m\to \infty$,
$$
\PP\left((X^K_{1/\sqrt{m}\,,n},Z^K_{1/\sqrt{m}\,,n})\, \in\, G_m \text{ for all } Tm > n > f(m) \right)\,\to\, 1
$$
Together with Proposition~\ref{prop-Y}
we see that the stopped processes $(X^K_{1/{\sqrt{m}},n},Z^K_{1/{\sqrt{m}},n})$ for $n=1,\ldots, mT$ satisfy the conditions of
Proposition~\ref{p_turboEK} with $\tilde X^m_n=P_1X^K_{1/{\sqrt{m}},n}$, the good sets $G_m$ and
$f(m)=\lfloor m^\alpha \rfloor$.
Thus, with Proposition~\ref{prop-P0X} it follows $X^K_{1/{\sqrt{m}},\lfloor tm \rfloor} \Longrightarrow \pmat{\nul &  \\  & \Lambda_t^K} X_0$, uniformly for
$0<t<T$, where $\Lambda_t^K = \Lambda_{t\wedge T_k}$ denotes the stopped process of
$\Lambda_t$ as described in Theorem~\ref{theo-SDE} with stopping time
$T_K$ when $\left\| P_1 \Lambda_t X_0 \right\| >K$.
As we have this convergence for all such stopping times $T_K$, $\|P_1 \Lambda_t X_0\|$ is almost surely finite
and as the final time $T$ was arbitrary, one obtains
$X_{1/\sqrt{m},\lfloor tm \rfloor} \Longrightarrow \pmat{\nul & \\  & \Lambda_t} X_0$ for any $t>0$.
Together with Proposition~\ref{prop-Z-small} this finishes the proof of Theorem~\ref{theo-SDE}.



\section{Application to random Schr\"odinger operators \label{sec-RSO}}

The main purpose of this section is the proof of Theorem~\ref{theo-GOE}. However, we will also obtain a description for limits of eigenvalue processes in a critical scaling.
For this we will consider slightly more general operators as \eqref{eq-H-lnd}.
More precisely, we study the limiting eigenvalue process for $n\to\infty$ with $\lambda^2n$ constant and $d$ fixed for more general random $nd\times nd$ matrices
given by
\begin{equation}\label{eq-def-H}
 (H_{\lambda,n} \psi)_k = \psi_{k+1}+\psi_{k-1}+(A+\lambda V_k)\psi_k\;
\end{equation}
where $\psi=(\psi_1,\ldots,\psi_n)$, $\psi_0=\psi_{n+1}=0$ and $\psi_j\in\CC^d$.
Here, $A$ is a general Hermitian matrix, and the $V_k$ are general i.i.d.\ Hermitian matrices with $\EE(V_k)=\nul$ and $\EE(\|V_k\|^{6+\varepsilon})<\infty$. 
We dropped the index $d$ now as $d$ will be fixed from now on and sometimes we may also drop the index $n$.
Moreover, for simplicity, we can assume that $A$ is diagonal; indeed, this can be achieved by the change of basis  
$\psi_n \mapsto O^* \psi_n$ where $O^* AO$ diagonalize $A$ (and replaces $V_n$ by $O^*V_n O$).

The eigenvalue equation $H_{\lambda}\psi=E \psi$ is a recursion that can be written in the matrix form as follows.
\begin{equation}\label{eq-def-transfer}
 \pmat{\psi_{k+1} \\ \psi_k} = T_{k} \pmat{\psi_k \\ \psi_{k-1}}\,\qquad
 \text{where}\quad
T_{k}= T^{E}_{\lambda,k} = \pmat{E\one -A-\lambda V_k & -\one \\ \one & \nul}\,.
\end{equation}
The $T^E_{\lambda,k}$ are called transfer matrices. Now, $E$ is an eigenvalue of $H_{\lambda,n}$ if there is a nonzero solution $(\psi_1,\psi_n)$ to
$$
 \pmat{0\\ \psi_n} = T_n\cdots T_{1} \pmat{\psi_1 \\ 0}\,\qquad
$$
or equivalently, when the determinant of the top left $d\times d$ block of $T_n\cdots T_{1}$ vanishes. So we can study the
eigenvalue equation through the products
\begin{equation*}
 T_{[1,k]} = T_k \cdots T_1
\end{equation*}
which are the focus of our next theorems.

\subsection{Elliptic and hyperbolic channels and SDE limits}

The matrices $T_k$ satisfy
$$T^*\Jj T = \Jj \qquad \text{where}\qquad  \Jj=\pmat{\nul & \one_d \\ -\one_d & \nul},$$
the definition of elements of the hermitian symplectic group $\HSp(2d)$.
In particular, they are all invertible. The $T_k$ are all perturbations of the noiseless matrix
$$T_*:=T^E_{0,1}\;.$$ This matrix is also block diagonal with $d$ blocks of size 2, and the eigenvalues of $T^E_{0,1}$ are exactly the
$2d$ solutions of the $d$ quadratics
$$
z+z^{-1}= E-a_j, \qquad a_j \text{ is an eigenvalue of } A.
$$
So the solutions are on the real line or on the complex unit circle, depending on
whether $|E-a_j|$ is less or more than two. We call the corresponding generalized eigenspaces of
$T_*=T^E_{0,1}$ elliptic $(<2)$, parabolic $(=2)$ and hyperbolic $(>2)$ channels. Elliptic and hyperbolic channels correspond to
two-dimensional eigenspaces, while parabolic channels correspond to a size 2 Jordan block. Traditionally, this notation refers to the solutions
of the noiseless $(\lambda=0)$ recursion that are supported in these subspaces for every coordinate $\psi_n$.

Pick an energy $E$, such that there are no parabolic channels and at least one elliptic channel.
Suppose that $A$ is diagonalized so that
$|E-a_j|>2$ for $j=1,\ldots,d_h$ and $|E-a_j|<2$ for $j>d_h$.
Correspondingly, we define the hyperbolic eigenvalues $\gamma_j$ and elliptic eigenvalues $z_j$ of $T_*$ by
\begin{align*}
 \gamma_j+\gamma_j^{-1} &= E-a_j\,,\quad &|\gamma_j|<1&\;,&\qtx{for}& j=1,\ldots,d_h \\
 z_j+z_j^{-1} &= E-a_{j+d_h}\,,\quad &|z_j|=1&,\,\im(z_j)>0\;,&\qtx{for}& j=1\,\ldots,d_e=d-d_h\;.
\end{align*}
Furthermore we define the diagonal matrices
\begin{equation}\label{eq-def-Ga-Z}
\Gamma=\diag(\gamma_1,\ldots,\gamma_{d_h}),\qquad
Z=\diag(z_{1},\ldots,z_{d_e}).
\end{equation}

In order to complete the description of the limiting eigenvalue process, we need to consider a family of limiting SDE by varying the energy in
the correct scaling. More precisely, define the $2d_e\times 2d_e$ unitary matrix $U$ and
the $2d \times 2d$ matrix $\Qq$ by
\begin{equation}\label{eq-def-Qq}
U=\pmat{\bar Z &  \\  & Z}\,,\qquad
 \Qq=\pmat{\Gamma &  &  & \Gamma^{-1} \\  & \bar Z & Z &  \\
 \one_{d_h} &  &  & \one_{d_h} \\  & \one_{d_e} & \one_{d_e} & }\,
\end{equation}
so that $\Qq$ diagonalizes $T_*$ to a form as in \eqref{eq-T_0} that is used for Theorem~\ref{theo-SDE}
\begin{equation*}
\Tt_*:=\Qq^{-1}T_{*}\Qq=\pmat{\Gamma &   &   \\   & U &   \\   &   & \Gamma^{-1}}\,.
\end{equation*}
Furthermore, let
\begin{equation}\label{eq-def-cjT}
\Tt_k= \Tt^{\varepsilon,\sigma}_{\lambda,k}:=
 \Qq^{-1}\, T^{E+\lambda^2\varepsilon}_{\lambda,k} \,\Qq=
 \Tt_*+\lambda \sigma\, \Vv_{k} +  \lambda^2 \varepsilon \Ww\,,\quad
 \Tt_{[1,k]}=\Tt^{\varepsilon,\sigma}_{\lambda,[1,k]}:=\Tt_{k}\cdots \Tt_{1}\,
\end{equation}
with
\begin{equation}\label{eq-T-parts}
\Vv_k=\Qq^{-1} \pmat{-V_k &  \\ & \nul} \Qq\,,\qquad \Ww= \Qq^{-1} \pmat{\one_d & \\ & \nul} \Qq\,.
\end{equation}
The parameter $\sigma$ is somewhat redundant, however, it will be useful for replicating the argument of \cite{VV1} where this scaling parameter was also introduced.
The scaling  $\varepsilon \lambda^2\sim \varepsilon/n$ means that a unit interval of  $\varepsilon$ should contain a constant order of eigenvalues.
In order to get limiting SDEs we consider a Schur complement as before, thus define the $2d_e \times 2d_e$ matrices
\begin{equation}\label{eq-def-hat-T}
 \wh \Tt^{\varepsilon,\sigma}_{\lambda,n} = 
 \left( \Pp^*_{\leq 1} \left[ \Tt^{\varepsilon,\sigma}_{\lambda,[1,n]} \Xx_0 \right]^{-1} \Pp_{\leq 1}\right)^{-1} \qtx{with}
 \Pp_{\leq1}=\pmat{\one_{d_h+2d_e} \\ \nul_{d_h\times(d_h+2d_e)}}\,
\end{equation}
Then by Theorem~\ref{theo-SDE} we obtain the correlated family (parameters $\sigma, \,\varepsilon$) of limiting processes
\begin{equation}\label{eq-L_et}
\pmat{ \nul_{d_h} & \\ & U^{-\lfloor tn \rfloor} }\,\wh \Tt^{\varepsilon,\sigma}_{1/\sqrt{n},\lfloor tn\rfloor}\;\Longrightarrow\;
 \pmat{ \nul_{d_h} \\ & \Lambda^{\varepsilon,\sigma}_{t}}\,\left( \Pp^*_{\leq 1} \Xx_0^{-1} \Pp_{\leq 1}\right)^{-1} \qquad \text{for $n\to\infty$}\;,
\end{equation}
where for fizzed $(\varepsilon, \sigma)$, the process $\Lambda^{\varepsilon,\sigma}_t$ satisfies some SDE in $t$.

\begin{remark}\label{rem-red-T}
For $\varepsilon=0$ and $\sigma=1$, up to some conjugation, the matrix $\widehat \Tt^{0,1}_{\lambda,n}$ corresponds to the reduced transfer matrix as introduced in {\rm \cite{Sa1}} for the scattering of a block described by $H_{\lambda}$ 
of a finite length $n$ inserted into a cable
described by $H_0$ of infinite length ('$n=\infty$'). Thus we obtain that in the limit $\lambda^2 n ={\rm const.},\,n\to\infty$, the process of the reduced transfer matrix as defined in {\rm \cite{Sa1}} is described by a SDE, proving Conjecture~1 in {\rm \cite{Sa1}}.
\end{remark}

To get to the GOE limit we need to express the limit SDEs more explicitly. Therefore, let us split the potential $V_1$ into the hyperbolic and elliptic parts, i.e. let
\begin{equation} \label{eq-def-Vh}
V_1  = \pmat{ V_h & V_{he} \\ V^*_{he} & V_e}\,\qtx{where}
 V_h\,\in\,\Mat(d_h\times d_h)\;,\quad V_e\,\in\,\Mat(d_e\times d_e)\;.
\end{equation}
Moreover, define
\begin{equation}\label{eq-def-Sh}
Q  = \int_{\langle Z \rangle} \zb \,\EE(V_{he}^* (\Gamma^{-1}-\Gamma)^{-1} \,V_{he})\,\bar \zb\,
 d \zb\;,\quad
 \Ss=\pmat{(\bar Z-Z)^{-1} & \nul \\ \nul & (\bar Z-Z)^{-1}}
\end{equation}
where $d\zb$ denotes the Haar measure on the compact abelian group $\langle Z \rangle$ generated  by the diagonal, unitary matrix $Z$. As we will see, $Q$ will give rise to a drift term coming from the hyperbolic channels. In fact, this is the only influence of the hyperbolic channels for the limit process.
Moreover, to simplify expressions, we will be interested in one specific case.
\begin{define}
We say that the matrix $Z=\diag(z_1,\ldots,z_{d_e})$ with $|z_j|=1, \,\im(z_j)>0$ is {\bf chaotic}, if all of the following apply for all $i,j,k,l \in \{1,\ldots, d_e\}$,
\begin{align*}
&z_i z_j z_k z_l \neq 1\;, \quad \bar z_i z_j z_k z_l \neq 1\,\\
&\bar z_i \bar z_j z_k z_l \neq 1  \quad \text{unless $\{i,j\}=\{k,l\}$}.
\end{align*}
\end{define}
The following observation corresponds to Lemma~8 in \cite{VV1}.
\begin{lemma}\label{lem-chaotic}
 Let the eigenvalues $a_j$, $j=1,\ldots, d$ of $A$ be simple and let $I$ be the interval with fixed hyperbolic and elliptic channels as considered, i.e.
 $$
 I\,=\,\{E\in\RR\,:\,|E-a_j|>2\;\; \text{for $j=1,\ldots, d_h$\;\; and}\;\; |E-a_j|<2\;\;\text{for $j>d_h$}\;\}\,.
 $$
 Then, for Lebesgue almost all $E\in I$, the matrix $Z$ as defined above is chaotic and moreover, for any diagonal, unitary matrix $Z_*$ there is a sequence $n_k$ such that
 $Z^{n_k+1}\to Z_*$.
\end{lemma}
\begin{proof}
 By the definitions above, $z_j=e^{i\varphi_j}$ where $\varphi_j = \arccos((E-a_{j+d_h})/2)\,\in\,(-\pi\,,\,\pi)$.
We will show that for almost all $E$, the vector $\varphi=\varphi(E)=(\varphi_{1},\ldots,\varphi_{d_e})$ has no non-zero integer vector orthogonal to it.
It is not difficult to see that $Z$ is chaotic in this case and the orbit $Z^n$ is dense in the torus of diagonal unitary matrices.\\
It is enough to show that for any non-zero integer vector $w$ the set of energies $E\in I$ where $w \cdot \varphi(E) = 0$ is finite.
Clearly, $E\mapsto w\cdot \varphi(E)$ is analytic on $I$ and therefore it either has finitely many zeros or is constant zero.
Taking the derivative with respect to $E$ we get
$$
(w\cdot \varphi(E))'\,=\,\sum_{j=1}^{d_e} \frac{-w_j}{\sqrt{1-\frac14(E-a_{j+d_h})^2}}\,.
$$
As all the values $a_{j+d_h}$ are different, each summand has a singularity at a different value. Hence, this derivative can only be identically zero on $I$ if
$w$ is the zero vector. Thus, for $w\neq 0$, $E\mapsto w\cdot \varphi(E)$ is not the zero function.
\end{proof}

\begin{prop}
\label{prop-SDEs}
{\rm (i)}
 The family of processes $\Lambda^{\varepsilon,\sigma}_{t}$ as in equation
 \eqref{eq-L_et} or Theorem~\ref{theo-EV}
 satisfy SDEs of the form
 \begin{equation}\label{eq-sigma-SDE}
 d\Lambda^{\varepsilon,\sigma}_{t}\,=\,\Ss
 \pmat{\varepsilon \one - \sigma^2Q & \\ & -\varepsilon \one + \sigma^2Q} \,\Lambda^{\varepsilon,\sigma}_{t}\,dt
 \,+\,\sigma\,\Ss \pmat{d\Aa_t & d\Bb_t \\ -d\Bb_t^* & -d\Cc_t}\,\Lambda^{\varepsilon,\sigma}_{t}\,
 \end{equation}
 with $\Lambda^{\varepsilon,\sigma}_{0}=\one$ and $\sigma,\varepsilon$ fixed,
where $\Aa_t,\,\Bb_t,\,\Cc_t$ are jointly Gaussian complex-valued $d_e\times d_e$ matrix Brownian motions, independent of $\varepsilon$ and $\sigma$, with $\Aa_t^* = \Aa_t$, $\Cc_t^*=\Cc_t$ and certain covariances.

\noindent {\rm (ii)} If $A$ and $V_n$ are real symmetric then we obtain
\begin{equation*}
 \Cc_t = \overline{\Aa_t}=\Aa_t^\top\qtx{and} \Bb_t^\top = \Bb_t\,.
\end{equation*}

\noindent {\rm (iii)} If $Z$ is chaotic then $\Bb_t$ is independent of $\Aa_t$ and $\Cc_t$. Also,
$\Aa_t$ and $\Cc_t$ have the same distribution.
Moreover, with the subscript $t$ dropped, we have the following:
\begin{align*}
&\EE|\Aa_{ij}|^2=\EE|\Bb_{ij}|^2=t\,\EE|(V_e)_{ij}|^2 \\
&\EE(\Aa_{ii} \Aa_{kk})=
t\,\EE((V_e)_{ii} (V_e)_{kk})\,, \\
&\EE(\Aa_{ij} \Cc_{ij}) = \EE(\Aa_{ij}\overline{\Cc}_{ji})=
\EE(\Bb_{ij}\overline{\Bb}_{ji})=t\,\EE((V_e)_{ij})^2
\end{align*}
and whenever $\{i,j\}\neq\{k,l\}$ one finds
\begin{equation*}
 \EE(\Aa_{ij}\Aa_{kl})=
\EE(\Aa_{ij} \Cc_{kl})=
\EE(\overline{\Bb}_{ij} \Bb_{kl})=0
\end{equation*}
and for any $i,j,k,l$,
\begin{equation*}
 \EE(\Bb_{ij} \Bb_{kl})=0\,.
\end{equation*}
All other covariances are obtained from $\Aa_t=\Aa_t^*,\,\Cc_t=\Cc_t^*$.
\end{prop}

\begin{proof} Let us stick to the case $\sigma=1$. 
Note that
\begin{equation}
 \Qq^{-1}=\Ss_{\Gamma,Z} \pmat{\one_{d_h} & \nul & - \Gamma^{-1} & \nul \\ \nul & \one_{d_e} & \nul & -Z \\
 \nul & -\one_{d_e} & \nul & \bar Z \\ -\one_{d_h} & \nul & \Gamma & \nul }\;
 \end{equation}
 where
 \begin{equation}
 \Ss_{\Gamma,Z}=\pmat{-S_\Gamma &  &  &  \\  & S_Z &  &  \\  &  & S_Z &  \\
  &  &  & -S_\Gamma}\,,\quad S_\Gamma=(\Gamma^{-1}-\Gamma)^{-1}\;,\quad
 S_Z=(\bar Z-Z)^{-1}
\end{equation}
we chose the sign on $S_\Gamma$ this way, so that $S_\Gamma>\nul$ is a positive diagonal matrix.
With \eqref{eq-T-parts} and \eqref{eq-def-Vh} this leads to
\begin{equation}
 \Vv_1=\Ss_{\Gamma,Z} \pmat{-V_h\Gamma & -V_{he} \bar Z & -V_{he} Z & -V_h \Gamma^{-1} \\
 -V_{he}^* \Gamma & -V_e \bar Z & -V_e Z & -V_{he}^* \Gamma^{-1} \\
 V_{he}^* \Gamma & V_e \bar Z & V_e Z & V_{he}^* \Gamma^{-1} \\
 V_h\Gamma & V_{he} \bar Z & V_{he} Z & V_h \Gamma^{-1}}\;,\quad
 \Ww=\Ss_{\Gamma,Z} \pmat{\Gamma & \nul & \nul & \Gamma^{-1} \\
 \nul & \bar Z & Z & \nul \\  \nul & -\bar Z & -Z & \nul \\ -\Gamma & \nul & \nul & \Gamma^{-1}}
\end{equation}
In the notations as introduced in Section~\ref{sec-intro} and used for Theorem~\ref{theo-SDE}
we have $\Gamma_2=\Gamma$ and
\begin{eqnarray}
 \vl11 U^* &=& \Ss \left(\pmat{-V_e & -V_e \\ V_e & V_e} + \lambda \epsilon \pmat{\one_{d_e} & \one_{d_e} \\ - \one_{d_e} & - \one_{d_e} }\right) \\
 \vl12\Gamma\vl21 U^* &=& \Ss \pmat{V_{he}^*S_\Gamma V_{he} & V_{he}^* S_\Gamma V_{he} \\
 -V_{he}^*S_\Gamma V_{he} & -V_{he}^* S_\Gamma V_{he}}\,
\end{eqnarray}
with $\Ss$ as in \eqref{eq-def-Sh}.
In order to calculate the drift term, note that using the definition of $Q$ in \eqref{eq-def-Sh} we obtain
$$
\int_{\langle Z \rangle} \pmat{\zb & \nul \\ \nul & \bar \zb} \,\EE\,\pmat{\varepsilon\one-V_{he}^* S_\Gamma V_{he} &
\varepsilon\one-V_{he}^* S_\Gamma V_{he}
 \\ -\varepsilon\one+V_{he}^* S_\Gamma V_{he} & -\varepsilon\one+V_{he}^* S_\Gamma V_{he}}
\pmat{\bar\zb & \nul \\ \nul & \zb}\,d\zb\,=\,
\pmat{\varepsilon\one-Q & \nul \\ \nul & -\varepsilon\one+Q}
$$
where we used that for any $d_e\times d_e$ matrix $M$ one finds
$$
\int_{\langle Z \rangle} \zb \,M\,\zb\,d\zb\,=\,
\lim_{n\to\infty} \frac1n \sum_{k=1}^n Z^k M Z^k =
\left(\lim_{n\to\infty} \frac1n\sum_{k=1}^n (z_i z_j)^k M_{ij}\right)_{ij}\,=\,\nul
$$
as we have $|z_iz_j|=1$ and $\im(z_i)>0,\,\im(z_j)>0$ implying that
$z_i z_j \neq 1$ for any $i,j \in \{1,\ldots,d_e\}$.
Therefore, application of Theorem~\ref{theo-SDE} gives \eqref{eq-sigma-SDE}
with $\Aa_t=\Aa_t^*$, $\Cc_t=\Cc_t^*$.
In order to express the covariances as described by \eqref{eq-variances} in more detail recall
$Z=\diag(z_1,\ldots,z_{d_e})$, $|z_j|=1$, leading to
\begin{align}
 \int_{\langle Z \rangle} \prod_{j=1}^{d_e} \zb_{jj}^{n_j} \,d\zb\,=\,
 \chi\left(\prod_{j=1}^{d_e} z_j^{n_j} \right)
 \qtx{with} \chi(z)=\begin{cases}
          1 & \qtx{for} z=1 \\
          0 & \qtx{else}
         \end{cases}\;
 \end{align}
 where $\zb_{jj}$ is the $j$-th diagonal entry of the diagonal matrix $\zb\in\langle Z \rangle$,
 and $n_j$ are integers.
This leads to the following covariances,
\begin{align*}
 \EE((\Aa_t)_{ij} (\Aa_t)_{kl}) &= \EE((\overline{\Aa_t})_{ji} (A_t)_{kl}) =\EE((\Cc_t)_{ij} (\Cc_t)_{kl}) = \notag  \EE((\overline{\Cc_t})_{ji} (\Cc_t)_{kl}) \\
 &= t\,\EE((V_e)_{ij}(V_e)_{kl}) \,\chi(\bar z_i z_j \bar z_k z_l)\;; \\
  \EE((\Bb_t)_{ij} (\Bb_t)_{kl}) &= t\,\EE((V_e)_{ij}(V_e)_{kl})\,\chi(z_iz_jz_kz_l)\;; \\
  \EE((\overline{\Bb_t})_{ij} (\Bb_t)_{kl}) &= t\, \EE((\overline{V_e})_{ij}(V_e)_{kl})\,
 \chi(\bar z_i \bar z_j z_k z_l)\;.
 \end{align*}

The correlations between the Brownian motions are given by
\begin{align*}
  &\EE((\Aa_t)_{ij} (\Cc_t)_{kl}) = \EE((\overline{\Aa_t})_{ji}(\Cc_t)_{kl})
  = t \,\EE((V_e)_{ij}(V_e)_{kl})\,\chi(z_i \bar z_j \bar z_k z_l)\;;\\
 &\EE((\Aa_t)_{ij}(\Bb_t)_{kl}) =
 \EE((\overline{\Aa_t})_{ji}(\Bb_t)_{kl}) = t\,\EE((V_e)_{ij}(V_e)_{kl})\,\chi(z_i\bar z_j z_k z_l)\;; \\
  &\EE((\Cc_t)_{ij}(\Bb_t)_{kl}) =
 \EE((\overline{\Cc_t})_{ji}(\Bb_t)_{kl}) = t\,\EE((V_e)_{ij}(V_e)_{kl})\,\chi(\bar z_i z_j z_k z_l) \,.
 \end{align*}
This shows part (i) for $\sigma=1$. Changing $V_1$ to $\sigma V_1$ immediately gives the general case. 
If $V_e$ is almost surely real, which is the case if $O^* V_1 O$ is almost surely real, then one has $\Cc_t=\overline{\Aa}_t$ and $\Bb_t = \Bb_t^\top$ giving part (ii).
Part (iii) follows from using the chaoticity assumption in the equations for the covariances.
\end{proof}

\subsection{Limiting eigenvalue statistics \label{sub:eigenvalue}}

The convergence to the SDE limit as in \eqref{eq-L_et} should firstly be interpreted for fixed $\varepsilon$ and $\sigma$.
However, considering direct sums of matrices for finitely many pairs $(\varepsilon, \sigma)$
one obtains 
joint convergence to a random field $(\varepsilon,\sigma,t)\mapsto \Lambda_t^{\varepsilon,\sigma}$ in terms of finite dimensional distributions. 
For fixed $\sigma,\,t$, the left hand side of  \eqref{eq-L_et} is clearly analytic in $\varepsilon\in\CC$.
Moreover, all estimates made for the general setup are uniform for $\varepsilon$ varying in compact sets.
Using the bounds \eqref{eq-Z-small} and \eqref{eq-XX-bound1} we can therefore apply \cite[Corollary 15]{VV1} and see that for fixed $\sigma$ and $t$ there is a unique\footnote{unique in the sense of a uniquely induced distribution on the set of analytic functions} version such that $\varepsilon \mapsto \Lambda^{\varepsilon,\sigma}_t$ is analytic.
In particular, using this analytic version, we can define the random set $\{\varepsilon \in \CC\;:\; f(\Lambda^{\varepsilon,\sigma}_t) = 0\}$ for fixed $(\sigma, t)$ and an analytic function
$f:\Mat(d,\CC)\to \CC$. Unless $f(\Lambda^{\varepsilon,\sigma}_t)$ is the zero function in $\varepsilon$, this random set consists of isolated points by analyticity and can be seen as a point process which we may denote by
${\rm zeros}_\varepsilon f(\Lambda^{\varepsilon,\sigma}_t)$.

\begin{theo}\label{theo-EV}
Consider the process $\Ee_{\sigma,n}$ of eigenvalues of
$n(H_{\frac\sigma{\sqrt{n}},n}-E)$ and let $n_k$ be an increasing sequence such that $Z^{n_k+1}\to Z_*$ for
$k\to\infty$ with $Z$ being the unitary, diagonal $d_e\times d_e$ matrix defined in \eqref{eq-def-Ga-Z}.
Then, $\Ee_{\sigma,n_k}$ converges to the zero process of the determinant of a $d_e\times d_e$ matrix,
\begin{equation*}
 \Ee_{\sigma,n_k}\,\Longrightarrow\, {\rm zeros}_\varepsilon \det\left( \pmat{\bar Z_*  &  Z_*}
 \Lambda^{\varepsilon,\sigma}_{1} \pmat{\one_{d_e} \\ - \one_{d_e}}
 \right)\;.
\end{equation*}
 \end{theo}

\begin{proof} Without loss of generality we restrict to the case $\sigma=1$.
We won't need the precise form of the limit SDE but it is important how we obtain this SDE.
Therefore we need to look at the matrix parts giving the Schur complement as in the proof of
Theorem~\ref{theo-SDE}. Hence, using $U$ and $\Tt^{\varepsilon,1}_{\lambda,[1,n]}$ as above let
\begin{equation*}
 \Rr=\pmat{\one_{d_h} &  &  \\  & U &  \\  &  & \one_{d_h}}\,,\quad
 \Xx_0=\pmat{\one_{d_h} & \nul & -\one_{d_h} \\ \nul & \one_{2d_e} & \nul \\ \nul & \nul & \one_{d_h}}
\end{equation*}
and
\begin{equation}\label{eq-def-Xx-e-l-n}
 \Xx^{\varepsilon}_{\lambda,n}=\Rr^{-n} \Tt^{\varepsilon,1}_{\lambda,[1,n]} \Xx_0\,.
\end{equation}
Using blocks of sizes $d_h+2d_e$ and $d_h$, let
\begin{equation}\label{eq-def-X-e-l-n}
\Xx^{\varepsilon}_{\lambda,n} =
\pmat{A^{\varepsilon}_{\lambda,n} & B^{\varepsilon}_{\lambda,n} \\ C^{\varepsilon}_{\lambda,n} & D^{\varepsilon}_{\lambda,n}}\;, \quad
X^{\varepsilon}_{\lambda,n}=A^{\varepsilon}_{\lambda,n}-B^{\varepsilon}_{\lambda,n}
(D^{\varepsilon}_{\lambda,n})^{-1} C^{\varepsilon}_{\lambda,n}\,.
\end{equation}
Then by Theorem~\ref{theo-SDE},
$X_{\varepsilon,1/\sqrt{n},\lfloor tn \rfloor} \Longrightarrow \pmat{\nul &  \\  & \Lambda^\varepsilon_{t} }$,
with the process $\Lambda^\varepsilon_t=\Lambda^{\varepsilon,1}_{t}$ as in \eqref{eq-L_et} and Theorem~\ref{theo-EV}.
Let us define
\begin{equation}
\Theta_0:=  \Xx_0^{-1} \Qq^{-1} \pmat{\one_{d} \\ \nul}\,\pmat{\nul & \Gamma^{-1}-\Gamma \\ (\bar Z-Z) & \nul}\,=\,
  \pmat{\nul & \nul \\ \one_{d_e} & \nul \\ -\one_{d_e} & \nul \\ \nul & \one_{d_h} }
\end{equation}
as well as
\begin{equation}
M^{\varepsilon}_{\lambda,n}=\pmat{\one_{d_e} & & \nul \\ - \left( D^{\varepsilon}_{\lambda,n}\right)^{-1} C^{\varepsilon}_{\lambda,n}
\smat{\nul \\ \one_{d_e} \\ -\one_{d_e}} & & ( D^{\varepsilon}_{\lambda,n})^{-1} } \,\in\,
\GL(d,\CC)
\end{equation}
Then, 
\begin{equation}\label{eq-Xx-Th-M}
 \Xx^{\varepsilon}_{\lambda,n} \,\Theta_0\, M^{\varepsilon}_{\lambda,n}
\,=\,
\pmat{X^{\varepsilon}_{\lambda,n} \smat{\nul \\ \one \\ -\one} & B^{\varepsilon}_{\lambda,n}(D^{\varepsilon}_{\lambda,n})^{-1} \\ \nul & \one_{d_h} }
\end{equation}


Let us also define
\begin{equation}\label{eq-def-Th-n}
 \Theta^*_n:=\pmat{\nul & \one_{d_e} \\ \Gamma^{-1} & \nul}\,\pmat{\one_{d} & \nul}\, \Qq \Rr^n\,=\,
 \pmat{\nul & {\bar Z}^{n+1} & Z^{n+1} & \nul \\ \one_{d_h} & \nul & \nul & \one_{d_h}}\,.
\end{equation}

An energy $E+\lambda^2\varepsilon$
is an eigenvalue of $H_{\lambda,n}$, precisely if there is a solution to the eigenvalue equation
with $\psi_{0}=0$ and $\psi_{n+1}=0$, i.e. if and only if
\begin{equation}\label{eq-eig-cond}
 \det \left(\pmat{\one_{d} & \nul} \Tt^{E+\lambda^2\varepsilon}_{\lambda,[1,n]}\,\pmat{\one_{d} \\ \nul}\right) = 0\,.
\end{equation}
As $\Tt^{E+\lambda^2\varepsilon}_{\lambda,[1,n]}=\Qq \Rr^n \Xx^{\varepsilon}_{\lambda,n}\Xx_0^{-1} \Qq^{-1}$, this is equivalent to
\begin{equation}
\det\left(\Theta_n^* \Xx^{\varepsilon}_{\lambda,n} \Theta_0  M^{\varepsilon}_{\lambda,n}\right)=0\,
\end{equation}

Using Theorem~\ref{theo-SDE}, \eqref{eq-Xx-Th-M} and \eqref{eq-def-Th-n}, we see that along a subsequence $n_k$ of the positive integers where $Z^{n_k+1}$ converges to $Z_*$, we find 
for $\lambda_k=1/\sqrt{n_k}$
\begin{equation}\label{eq-TXTM-conv}
\Theta_{n_k}^* \Xx^\varepsilon_{\lambda_k,n_k} \Theta_0  M^\varepsilon_{\lambda_k,n_k} \,\Longrightarrow\,
\pmat{ \pmat{\bar Z_* &  Z_*} \Lambda^\varepsilon_1 \pmat{\one_{d_e} \\ -\one_{d_e}} & & \nul \\ \nul & & \one_{d_h}}\;.
\end{equation}
As already established above, there is a unique holomorphic version of the random process $\varepsilon\mapsto  \Lambda_{\varepsilon,1}$. 
In fact, using uniform boundedness of $\Theta_n$ as well as \eqref{eq-Xx-Th-M} and the bounds \eqref{eq-Z-small}, \eqref{eq-XX-bound1} we find that
$\EE\| \Theta_{n_k}^* \Xx_{\varepsilon,\lambda_k,n_k} \Theta_0  M_{\varepsilon,\lambda_k,n_k} \|$ is uniformly bounded for $\varepsilon \in \CC$ varying in a compact set.
Hence, we can apply \cite[Corollary 15]{VV1} to obtain again the existence of a unique analytic version in $\varepsilon$ for the right hand side.
Moreover, we find a realization of all random processes on the same probability space but with local-uniform convergence in $\varepsilon$ in the equation
\eqref{eq-TXTM-conv} (Skorokhod embedding theorem).
As the determinant is a holomorphic function, the same is true for the determinants of these matrices.
As long as the (random) holomorphic determinant of the right hand side is not identically zero the local uniform convergence also implies that the discrete level sets of zeros of the determinants converge in the vague sense,
i.e. the counting measures integrated against continuous, compactly supported functions converge.
It is possible that certain zeros go off to infinity and disappear in the limit.

Hence, it is left to show that $\det(\smat{\bar Z_* &  Z_*} \Lambda^\varepsilon_1 \smat{\one \\ -\one}) $ is (almost surely) not identically zero in $\varepsilon$.
Now, \eqref{eq-sigma-SDE} can be rewritten as
$$
\pmat{\one & \\ & -\one} \Ss^{-1}\,d\Lambda^\varepsilon_t\,+\,\pmat{Q \\ & Q}\,\Lambda^\varepsilon_{t}\, dt\,+\,
\pmat{d\Aa_t & d\Bb_t \\ d\Bb^*_t & -d\Cc_t}\,\Lambda^\varepsilon_{t}\,=\,\varepsilon\,\Lambda^\varepsilon_{t}\,dt
$$
which is the transfer matrix equation (fundamental solution) for the eigenvalue equation $\Dd \psi = \varepsilon \psi$ where $\Dd$ is the random operator
$$
\Dd \psi(t) = \left[ \pmat{\one \\ & -\one} \Ss^{-1} \partial_t\,+\,\pmat{Q\\ & Q} + \pmat{d\Aa_t & d\Bb_t \\ d\Bb^*_t & -d\Cc_t}\,/\,dt \right]\,\psi(t)\,
$$
Using the continuous versions of the Brownian motions leading to measure valued white noise, one can make perfect sense of this random operator $\Dd$ on $L^2([0,1])\otimes \CC^{2d_e}$, by choosing 
the random domain of continuous functions $\psi(t)$ such that $\Dd \psi(t)$ (at first defined as a measure) is a continuous function.
(A typical procedure for first-order one-dimensional operators with measure-valued potential).

The zero determinant condition above yields an eigenvector $\psi$ satisfying the boundary conditions $\psi(0)=\smat{\one \\ -\one} \psi_1$ (i.e. $\pmat{\one & \one} \psi(0) = 0$) and 
$ \pmat{\bar Z_* & Z_*} \psi(1) = 0$. One can check that the operator is symmetric with these boundary conditions. Indeed, using integration by parts one finds for continuous 
$\psi(t),\,\varphi(t)$ in the domain with these boundary conditions, that
\begin{align*}
& \int_0^1 (\Dd \psi(t))^* \varphi (t)\,dt\,-\,\int_0^1 \psi^*(t) \Dd \varphi (t)\,dt \,=\,- \left[\psi^*(t)   \smat{S_Z^{-1} & \\ & -S_Z^{-1}} \varphi(t) \right]_0^1 \\
& \,=\, \psi^*(0) \smat{\one \\ \one} S_Z^{-1} \smat{\one & \nul} \varphi(0)\,+\, \psi^*(1) \smat{Z_* \\ \bar Z_*} S_Z^{-1} \smat{\nul & Z_*} \varphi(1)\,=\,0
\end{align*}
In the second line we used the boundary conditions first for $\varphi$ and then for $\psi$.
Hence, the set of eigenvalues $\varepsilon$ of $\Dd$ with these boundary conditions is a subset of the real line and in fact discrete and it is equal to the
zero set in $\varepsilon$ of the right hand side
of \eqref{eq-TXTM-conv}.
\end{proof}

\subsection{Limiting GOE statistics}



In this subsection we will prove Theorem~\ref{theo-GOE} by reduction to the work in \cite{VV1}.
Without loss of generality we focus on energies $E$ smaller than $0$ and consider $r=1$.
The more general case needs some more care and notations in the subdivision into elliptic and hyperbolic channels, but the main calculations remain the same.
We need to consider the SDE limit as described above a bit more precisely for this particular Anderson model as in \eqref{eq-def-H} with $A=\ZZ_d$ and $V_n$ as in \eqref{eq-GOE-cond1}.

In Proposition~\ref{prop-SDEs}, especially for the definitions of $V_h$, $V_e$ and $V_{he}$ it was assumed that $A$ is diagonal. So in order to use the calculations above we need to
diagonalize $\ZZ_d$ and see how this unitary transformation changes $V_n$.

We let $d\geq 2$, then $\ZZ_d$ is diagonalized by the orthogonal matrix $O$ given by
\begin{equation}
 O_{jk}=\sqrt{2/(d+1)}\,\sin(\pi jk /(d+1))\,.
\end{equation}
The corresponding eigenvalue of $\ZZ_d$ with eigenvector being the $j$-th column vector of $O$
is given by
\begin{equation}
a_j=2 \cos(\pi j / (d+1))\,,\quad j=1,\ldots, d\;.
\end{equation}
For $-2<E<0$ there is $d_h < d$ such that
\begin{align}
 2\cos(\pi j / (d+1))- E &> 2\qtx{for} j=1,\ldots,d_h \quad \text{and} \\
 -2< 2\cos(\pi j / (d+1))- E &< 2 \qtx{for} j=d_h+1,\ldots,d\,.
\end{align}
So we have $d_h$ hyperbolic and $d_e=d-d_h$ elliptic channels and the upper $d_h\times d_h$ block of $O^* \ZZ_d O$ corresponds to the hyperbolic channels.
Using \eqref{eq-GOE-cond1} and the notations as in \eqref{eq-def-Vh} we have
\begin{equation}
 \pmat{V_h & V_{he} \\ V_{he}^* & V_e} =
 O^\top \pmat{v_1 & & \nul \\ & \ddots & \\ \nul & & v_d}\,O\;,
 \quad \EE(v_j)=0\,,\;\; \EE(v_j\,v_k) = \delta_{jk}\,.
\end{equation}
Let $E$ be such that $Z$ is chaotic, then by Proposition~\ref{prop-SDEs}~(iii) we need to consider the following the covariances
\begin{align}
 \EE(|(V_e)_{ij}|^2)&=\EE\,\left|(O^\top V_1 O)_{i+d_h,j+d_h}\right|^2 = \langle |O_{i+d_h}|^2, |O_{j+d_h}|^2\rangle \\
 \EE((V_e)_{ii}(V_e)_{jj}) &=\EE\,\left((O^\top V_1 O)_{i+d_h,i+d_h}(O^\top V_1 O)_{j+d_h,j+d_h}\right) =
 \langle |O_{i+d_h}|^2, |O_{j+d_h}|^2\rangle\,.
\end{align}
Here, by $|O_i|^2$ we denote the vector $(|O_{k,i}|^2)_{k=1,\ldots,d}$ and $\langle\cdot,\cdot\rangle$
denotes the scalar product.
As stated in \cite{VV1}, one finds
\begin{equation}
(d+1)\; \langle |O_i|^2\,,\,|O_j|^2\,\rangle\, =\,
 \begin{cases}
     3/2 &\qtx{for} i=j \\
      1 &\qtx{for} i\neq j\;.
   \end{cases}
\end{equation}
Let us further calculate the drift contribution $Q$ from the hyperbolic channels as introduced above.
Using chaoticity, it is not hard to see from \eqref{eq-def-Sh} that $Q$ is diagonal. Moreover one has
\begin{equation}
 Q_{jj}=\EE(V_{he}^* S_\Gamma V_{he})_{jj})= \sum_{k=1}^{d_h} (S_{\Gamma})_{kk} \EE ([(V_{he})_{kj}]^2)=
 \sum_{k=1}^{d_h} \frac{\langle |O_k|^2, |O_{j+d_h}|^2\rangle}{\gamma_k^{-1}-\gamma_k}\;.
\end{equation}
It follows that $Q$ is a multiple of the unit matrix, more precisely
\begin{equation}
 Q=q\,\one\qtx{with} q = \frac{1}{d+1}\,\sum_{k=1}^{d_h} (\gamma_k^{-1}-\gamma_k)^{-1}\,.
\end{equation}
Note that $|q|<\max_k |\gamma_k^{-1}-\gamma_k|<\max_k |E-a_k| = \|E-A\|=\|E-\ZZ_d\|\leq |E|+2$ uniformly.
Thus, using Proposition~\ref{prop-SDEs} 
we obtain the following SDE limits,
\begin{equation}
 d\Lambda^{\varepsilon,\sigma}_{t}\,=\,
 \Ss\,(\varepsilon-\sigma^2q)\pmat{\one & \nul \\ \nul & -\one}\,\Lambda^{\varepsilon,\sigma}_{t}\,dt\,+\,
 \sigma \Ss \pmat{d\Aa_t & d\Bb_t \\ -d\Bb_t^* & -d\overline{\Aa}_t}\,\Lambda^{\varepsilon,\sigma}_{t}
\end{equation}
where $\Aa_t$ and $\Bb_t$ are independent matrix Brownian motions, $\Aa_t$ is Hermitian, $\Bb_t$ complex symmetric, i.e.
\begin{equation}
 \Aa_t^* = \Aa_t\;,\quad \Bb_t^\top = \Bb_t\;
\end{equation}
with covariance structure
\begin{equation}
\EE(|(\Bb_t)_{ij}|^2)=\EE(|(\Aa_t)_{ij}|^2)=\EE((\Aa_t)_{ii}(\Aa_t)_{jj})=
\begin{cases}
      \frac32\,t \,/\, (d+1) &\;\text{for}\; i=j\\
      t \,/\, (d+1)&\;\text{for}\; i\neq j
\end{cases}\;.
\end{equation}
All covariances which do not follow are zero. Except for the additional drift $\sigma^2 q$ which can be seen as a shift in $\varepsilon$, this is the exact same SDE as it appears in \cite{VV1}. 
In fact, the matrix $\Ss$ here corresponds to $iS^2$ as in \cite{VV1} and the process there corresponds to the process above conjugated by $|\Ss|^{1/2}$.\\
Thus, from now on the proof to obtain the ${\rm Sine}_1$ kernel and GOE statistics follows precisely the arguments as in \cite{VV1}.

First take $E$ as in Lemma~\ref{lem-chaotic} so that $Z$ is chaotic, and take a sequence $n_k$ such that $Z^{n_k+1} \to \one$, then, for the point process as in Theorem~\ref{theo-EV} we find
$\Ee_{\sigma,n_k} \Longrightarrow \Ee_\sigma={\rm zeros}_\varepsilon \det(\smat{\one & \one} \Lambda^{\varepsilon,\sigma}_1 \smat{\one \\ -\one})$\,.
Defining $\widehat\Lambda^{\varepsilon,\sigma}_t=\sigma^{-1} (\Lambda^{\varepsilon \sigma,\sigma}_t - \one)$ we find 
$$\sigma^{-1} \Ee_\sigma \,=\, {\rm zeros}_\varepsilon
\det\left(\pmat{\one & \one} \widehat \Lambda^{\varepsilon,\sigma}_1 \pmat{\one \\ -\one}\right)$$
where $\widehat \Lambda^\varepsilon_0=\nul$ and
$$
d \widehat \Lambda^{\varepsilon,\sigma}_t\,=\,(\varepsilon-\sigma q) \Ss \pmat{\one \\ & -\one} (\sigma \widehat \Lambda^{\varepsilon,\sigma}_t + \one)\,dt\,+\,
\Ss \pmat{d\Aa_t & d\Bb_t \\ -d\Bb_t^* & -d\overline{\Aa}_t}\,(\sigma \widehat\Lambda^{\varepsilon,\sigma}_{t} + \one)
$$
By \cite[Theorem 11.1.4]{SV} this SDE converges for $\sigma\to 0$ to the solution of the SDE with $\sigma=0$ which is a matrix-valued Brownian motion with drift and explicitly solvable.
Thus, for $\sigma \to 0$ one has $\widehat \Lambda^{\varepsilon,\sigma}_t \Longrightarrow \widehat \Lambda^{\varepsilon}_t$ which satisfies the same SDE with $\sigma=0$ above, therefore
$$
\widehat \Lambda^{\varepsilon,\sigma}_t \;\stackrel{\sigma\to 0}{\Longrightarrow}\; \widehat \Lambda^{\varepsilon}_t
\,=\, \varepsilon t\, \Ss \pmat{\one \\ & -\one}\,+\,\Ss\,\pmat{\Aa_t & \Bb_t \\ -\Bb_t^* & -\overline{\Aa}_t}\;.
$$
Using analytic versions in $\varepsilon$ one obtains by similar arguments as above that
$$
\sigma^{-1} \Ee_\sigma\;\stackrel{\sigma\to 0}{\Longrightarrow}\; {\rm zeros}_\varepsilon \det\left(\pmat{\one & \one} \widehat \Lambda^\varepsilon_1 \pmat{\one \\ - \one} \right)\,=\,
{\rm spec}\,(\re(\Bb_1-\Aa_1))
$$
where ${\rm spec}(\cdot)$ denotes the spectrum and $\re(\cdot)$ the entry-wise real part of a matrix. The latter equation is a simple calculation using the relations from above.
Similar to Proposition~9 in \cite{VV1}, the convergence can be realized jointly.
\begin{lemma}\label{lem-conv-RM}
Let $Z$ be chaotic, let $n_k$ be a sequence such that $Z^{n_k+1}\to \one$ and let  $\sigma_k$ be sequence with $\sigma_k\to 0$ such that $\frac{1}{\sigma_k} \| Z^{n_k+1} - \one \| \to 0$.
Consider the regularized Schur complements $\widehat \Tt^{\varepsilon,\sigma}_{\lambda,n}$ of the transfer matrices as defined in \eqref{eq-def-hat-T}.
Then define the regularized versions $\Xx^{\varepsilon,\sigma}_{\lambda,n}$ and the part $X^{\varepsilon,\sigma}_{\lambda,n}$ as in \eqref{eq-def-Xx-e-l-n} and \eqref{eq-def-X-e-l-n} but this time keeping the $\sigma$.
Choose $\Xx_0$ such that the corresponding Schur complement $X_0$ exists and let $\hat X_0=\smat{\nul \\ & \one} X_0$, the starting point for the SDE limit. Then, for $t>0$,
we find for $k\to \infty$ that
$$
\frac{1}{\sigma_k} \left( X^{\varepsilon \sigma_k, \sigma_k}_{\frac{1}{\sqrt{n_k}},\lfloor tn_k\rfloor}\;-\; \widehat X_0 \right) \quad \Longrightarrow \quad \pmat{\nul \\ & \widehat \Lambda^{\varepsilon}_t} \widehat X_0
$$
jointly for $t\in [0,1]$ and $\varepsilon$ varying in any finite subset of $\CC$.
Moreover, for the eigenvalue process $\Ee_{\sigma_k, n_k, d}$ of $(H_{\lambda_k,n_k,d}-E)$ we find
$$
\frac{n_k}{\sigma_k}\; \Ee_{\sigma_k, n_k, d}\;\Longrightarrow \; {\rm spec}\,\left(\re(\Bb_1-\Aa_1)\right)
$$
\end{lemma}
\begin{proof}
The proof for the first statement works very similar to above using Proposition~\ref{p_turboEK}. 
Therefore we let $\sigma_\lambda \to 0$ for $\lambda \to 0$ with $\sigma_{\lambda_k}=\sigma_k$ for $\lambda_k=1/\sqrt{n_k}$ and consider the process 
$$
\widehat X_{\lambda,n}\,=\, \frac{1}{\sigma_\lambda} \left[ X^{\varepsilon \sigma_\lambda, \sigma_\lambda}_{\lambda, n}\,-\,\hat X_0 \right]\,.
$$
For $\sigma_\lambda \widehat X_{\lambda,n}+\widehat X_0 = X^{\varepsilon \sigma_\lambda, \sigma_\lambda}_{\lambda, n}$ the drift term for each step is of order $\lambda^2 \varepsilon \sigma_\lambda$ and the diffusion term of order $\lambda \sigma_\lambda$.
Similar to \eqref{eq-exp-X} one obtains therefore an equation of the form
$$
\widehat X_{\lambda,n}\,=\,\pmat{\Gamma \\ & \one} \widehat X_{\lambda,n}\,+\,
\lambda R^{-n} V^X_{\lambda,n} R^{n-1} (\sigma_\lambda \widehat X_{\lambda,n}+\widehat X_0 )\,+\,\Oo\;.
$$
where the second term of $\vx$ in \eqref{eq-def-W_lb} gets an additional $\sigma_\lambda$ factor and the drift component of 
the first term, $\va$, is proportional to $\varepsilon \lambda$.
As $\sigma_\lambda\to 0$ the estimates on the reminder terms improve and the drift and diffusion terms will not depend on  $X_{\lambda,n}$ in the limit anymore.
Therefore, we get the Brownian motion with drift,
$\widehat X_{1/\sqrt{n},\lfloor tn \rfloor }\;\Longrightarrow\; \smat{\nul \\ & \widehat \Lambda^\varepsilon_t} \widehat X_0$\,.

To see the convergence of the eigenvalue processes we need to follow the calculations of Section~\ref{sub:eigenvalue} and use the analytic version with uniform convergence for compacts in $\varepsilon$.  Note that for this case $\widehat X_0=\smat{\nul \\ & \one}$ in blocks of sizes $d_h$ and $2d_e$.
With similar notations as in Section~\ref{sub:eigenvalue} (but keeping the $\sigma$-dependence and the upper $\sigma$-index) we obtain with $\sigma=\sigma_\lambda$ that
$$
\Theta_n^* \Xx^{\varepsilon \sigma, \sigma}_{\lambda,n} \Theta_0 M^{\varepsilon\sigma,\sigma}_{\lambda,n}\,=\,
\pmat{ \smat{\nul & \bar Z^{n+1} &  Z^{n+1}} (\sigma \widehat X_{\lambda,n}+\widehat X_0) \smat{\nul \\ \one \\ -\one} & & 
\smat{\nul & \bar Z^{n+1} & Z^{n+1} } \widehat Z_{\lambda,n}  
\\ \smat{\one & \nul & \nul} \sigma \widehat X_{\lambda,n} \smat{\nul \\ \one \\ -\one}  & & \one_{d_h}}\;
$$
where $\widehat Z_{\lambda,n}=B_{\lambda,n}^{\varepsilon \sigma,\sigma} (D_{\lambda,n}^{\varepsilon \sigma,\sigma})^{-1}$.
Note that by the choice of $\sigma_k$ as above one has
$$\sigma_k^{-1} \pmat{\nul & \bar Z^{n_k+1} &  Z^{n_k+1}} \widehat X_0 \smat{\nul \\ \one \\ -\one}\,=\,
\sigma_k^{-1}\,(\bar Z^{n_k+1}-Z^{n_k+1})\,\to\, 0\,.$$
Hence, for $\lambda_k=1/\sqrt{n_k}$ we have
$$
\Theta_n^* \Xx^{\varepsilon \sigma_k, \sigma_k}_{\lambda_k,n_k} \Theta_0 M^{\varepsilon\sigma_k,\sigma_k}_{\lambda_k,n_k} \pmat{\sigma_k^{-1} \one \\ & \one}
\;\;\Longrightarrow\;\; \pmat{ \smat{\one & \one} \widehat \Lambda^\varepsilon_1 \smat{\one \\ -\one} & \nul \\ \nul & \one_{d_h} }\,.
$$
Using analytic versions with uniform convergence locally in $\varepsilon$, the zero processes in $\varepsilon$ of the determinants also converge, hence
$\frac{n_k}{\sigma_k}\; \Ee_{\sigma_k, n_k, d}\;\Longrightarrow \;{\rm spec}\left(\re(\Bb_1-\Aa_1)\right)$.
\end{proof}

\begin{proof}[Proof of Theorem~\ref{theo-GOE}]
We still restrict to the case $r=1$. For any energy $E\in (-4,4)$, the number of elliptic channels $d_e=d_e(d)$ for the transfer matrices of $H_{\lambda,n,d}$ will go to $\infty$ as $d\to\infty$.
By Lemma~\ref{lem-chaotic} we find for Lebesgue almost all such energies that the following two things hold: \\
1. For any $d$ there is no parabolic channel (i.e. $|E-a_j|\neq 2$ for all $j$)\\
2. For any $d$ the conditions of Lemma~\ref{lem-chaotic} apply.\\ 
Take such an energy $E$ and take sequences $n_k=n_k(d)$ and 
$\sigma_k=\sigma_k(d)$ satisfying the conditions of Lemma~\ref{lem-conv-RM}.

$\re(\Bb_1-\Aa_1)$ is a real, symmetric $d_e\times d_e$ random matrix whose distribution depends only on $d_e$.
As noted in \cite[Section 4]{VV1} its distribution can be written as $(d+1)^{-1/2} (K+b\one)$ where $b$ is a standard Gaussian random variable and $K$ an independent real symmetric
matrix with mean zero and Gaussian entries such that $\EE(K_{ii}^2)=5/4$ and $\EE(K_{ij}^2)=1$ for $i\neq j$.
As explained in \cite{VV1} the bulk eigenvalue process $s(d_e)$ of $\sqrt{d_e} (K+b\one)$ converges locally to the ${\rm Sine}_1$ process by methods of \cite{ESYY} when $d_e$ converges to $\infty$.

Thus, for the eigenvalue processes $\Ee_{\sigma,n,d}$ of $H_{\sigma/\sqrt{n},n,d}-E$ we find
$ \frac{ \sqrt{ d\, d_e}\, n_k(d)}{\sigma_k(d)}\;\Ee_{\sigma_k,n_k,d}\,\Rightarrow\, s(d_e)$ and $s(d_e)\Rightarrow {\rm Sine}_1$ in the topology of weak convergence.
Thus we find some diagonal sequence $(k_j,d_j)$ such that with  $n_j=n_{k_j}(d_j),\,\sigma_j=\sigma_{k_j}(d_j),\,d_{j,e}=d_e(d_j)$ one finds
$ \frac{\sqrt{d_jd_{j,e}}\, n_j}{\sigma_j}\,\Ee_{\sigma_j,n_j,d_j}\,\Rightarrow\, {\rm Sine}_1$.
\end{proof}

\appendix

\section{Correlations along different directions and SDE limit on the flag manifold \label{sec-corr}}

Let $\Tt_0$ have eigenvalues of absolute value $c$ different from $1$, and $\Tt_0$ is diagonalized (or in Jordan form)
so that the corresponding eigenspace are also the span of coordinate vectors and have no Jordan blocks.
Then applying Theorem~\ref{theo-SDE} to the products of $\Tt_{\lambda,n}/c$ gives another SDE limit. Moreover, the convergence in law holds jointly for the
processes corresponding to magnitudes $1$ and $c$ (and in fact all magnitudes). 
Let us specify the covariance structure of the driving matrix-valued Brownian motions for the different processes.
Towards this, we define
$$
h_{1c}(M):=\lim_{\lambda\to 0} \EE(\vl11^\top M \vl11^{(c)})\;,\quad
\wh h_{1c}(M):=\lim_{\lambda\to 0} \EE(\vl11^* M \vl11^{(c)})
$$
where now $M$ is a $d_1(1)\times d_1(c)$ matrix, where $d_1(c)$ is the total dimension of all eigenspaces corresponding to eigenvalues of absolute value $c$.
$\vl11^{(c)}$ denotes the corresponding $d_1(c)\times d_1(c)$ block of $V_{\lambda,n}$.
Similarly, we define $h_{cc'}$ and $\wh h_{cc'}$ for any two absolute values $c,c'$ (see also \eqref{eq:h}).
As before, we also need the $d_1(c)\times d_1(c)$ unitaries $U_c$ (like $U$ in \eqref{eq-T_0}) so that $T_0$ restricted to
the eigenspaces of magnitude $c$ acts like $c U_c$.

\begin{theo}\label{theo-SDE2}
The convergence of Theorem \ref{theo-SDE} holds jointly along all eigenspaces corresponding to
absolute values $c$ of eigenvalues of $\Tt_0$ that correspond to eigenspaces without Jordan block.
We will denote the corresponding process for the magnitude $c$ by $\Lambda^{(c)}_t$.
Then, the covariance of the driving Brownian motions $\Bb,\Bb'$ for the magnitudes $c,c'$ are given by
\begin{equation} \label{eq-variances2}
\EE(\Bb_t^\top M \Bb'_t)=
g_{cc'}(M) t\,,\quad
\EE(\Bb_t^* M \Bb'_t)=\wh g_{cc'}(M) t
\end{equation}
where
\begin{align} \label{eq-def-gc}
 g_{cc'}(M) &= \frac{1}{cc'}\int_{\langle U_c,U_{c'}\rangle} \ov u\,\ov U_c \,h_{cc'}(u^\top M v)\,U_{c'}^* v^*\,d(u,v), \;\\
 \label{eq-def-hgc}
 \wh g_{cc'}(M) &=
 \frac{1}{cc'}\int_{\langle U_c,U_{c'}\rangle} u U_c\, \wh h_{cc'}(u^* M v)\, U_{c'}^* v^*\,d(u,v)\;.
\end{align}
Here, $\langle U_c,U_{c'}\rangle$ denotes the (block diagonal) compact abelian group generated by $\smat{U_c & \\ & U_{c'}}$,
and $d(u,v)$ denotes the Haar measure on $\langle U_c,U_{c'} \rangle \ni \smat{u & \\ & v}$.
\end{theo}

\begin{proof}
Consider the products of the direct sums 
$\widetilde \Tt_{\lambda,n}:=\Tt_{\lambda,n}/c \oplus \Tt_{\lambda,n}/c' = \pmat{\Tt_{\lambda,n} / c & \\ & \Tt_{\lambda,n} / c'}$.
In an adequate basis we can apply Theorem~\ref{theo-SDE} directly with $U$ being replaced by $\widetilde U = U_c \oplus U_{c'}$
and function $h$ being replaced by $\widetilde h \smat{M_0 & M_1 \\ M_2 & M_3}= \smat{h_{cc}(M_0)/c^2 & h_{cc'}(M_1) / (cc') \\ h_{cc'}(M_2) / (cc')  & h_{c'c'}(M_3) / c'^2}$\,.
A similar equation holds for the replacement of $\hat h$.
Then Theorem~\ref{theo-SDE} leads directly to the given statement.
\end{proof}



If the eigenvalues of $\Tt_0$ are of different absolute value, then the matrix product process grows at
different directions at different exponential rates. Hence there is no hope to get a matrix limit of the process
that captures all the directions and all the different SDE limits $\Lambda^{(c)}_t$ with their covariances at the same time.

First consider powers of the matrix $\Tt_0$ in the case it is diagonalizable and all the eigenvalues are of different positive absolute value.
Then high powers of $\Tt_0$ take most vectors close to the direction of the top eigenspace. A natural way to understand the second eigenvector through
typical behavior is through the action of $\Tt_0$ on two-dimensional subspaces. A high power of $\Tt_0$ will take two-dimensional eigenspaces into a two-dimensional space spanned by the top two eigenvectors of $\Tt_0$.

A flag is a nested sequence of subspaces of all dimensions up to $d$. The set of all such flags forms a compact manifold.
By the above argument, a high power of $\Tt_0$ takes most flags close to the flag given by
the nesting of the subspaces spanned by the top $k$ eigenvectors.

The above picture still holds when we add perturbations and consider the products
$\Tt_{\lambda,n} \cdots \Tt_{\lambda,1}$. So nothing interesting happens in this case. Things become more interesting when there are more than one eigenvalue of $\Tt_0$ for a given absolute value. If this holds for the top one, then the direction of the action of a typical vector becomes dependent on the randomness, even in the limit. The deterministic dynamics only gives that the vector will be in the subspace spanned by the eigenvectors corresponding to the top absolute value.
In this sense, the different exponential rates will still determine certain subspaces of the flag in the limit so that the limiting process will be in a specific submanifold that is invariant and attracting under the action of $\Tt_0$.

Our next theorem shows how this happens. More precisely, we will consider a flag which is typical for the behavior of powers of $\Tt_0$. This happens if the  $k$-dimensional spaces of the flag do not include directions that are spanned by subsets of eigenvectors of $\Tt_0$ corresponding to eigenvalues of lower order.  The matrix products applied to this flag will give a flag-valued process. This is described in Theorem~\ref{th-flag}.

As only invertible matrices act on a flag, suppose that for small $\lambda$ all $\Tt_{\lambda,n}$
are invertible with probability one, i.e., there is $\lambda_0$ such that for all $0\leq \lambda<\lambda_0$,
$\PP(\Tt_{\lambda,n}\; \text{is invertible for all $n$})=1$. Suppose further
that $\Tt_0$ is diagonalizable and that we chose a basis such that
\begin{equation}\label{eq-T0-diag}
 \Tt_0=\pmat{c_1 U_{c_1} & & \nul \\ & \ddots & \\ \nul & & c_k U_{c_k}}\,,\quad\text{where}
 \quad 0<c_1<c_2<  \ldots < c_k\;,
\end{equation}
with the $U_{c_j}$ being unitary $d(c_j) \times d(c_j)$ matrices.

A flag can be represented by an invertible $d\times d$ matrix $\Ff$ where the {\it last} $p$ column vectors, denoted by $\Ff^{(p)}$, span the $p$-dimensional subspace.
$\Ff_1$ and $\Ff_2$ represent the same flag if and only if $\Ff_1=\Ff_2 M$ for an invertible lower triangular matrix $M$. This forms an equivalence relation and we denote the equivalence class of $\Ff$ by $[\Ff]$.
Denoting the group of invertible, lower triangular $d\times d$ matrices by $\Delta(d)$ the flag manifold has
$$\FF=\GL(d,\CC)\,/\,\Delta(d).$$
The stable submanifold $\FF^s$ is the set of all flags such that the $d(c_1)+\ldots+d(c_j)$ dimensional subspace is spanned by the last $d(c_1)+\ldots+d(c_j)$ vectors in the standard basis, i.e.
\begin{equation*}
 \FF^s=\left\{ \left[ \smat{a_1 & & \nul \\ & \ddots & \\ \nul & & a_k } \right]\,:\, \text{for all $j$,}\;
a_j \in \GL(d(c_j))\,\right\}\,\subset\,\FF.
\end{equation*}
This is an attractor by the deterministic dynamics given by the action of $\Tt_0$ and the set of points in $\FF$ that is attracted is given by
\begin{align}
\FF^a = \left\{ \left[ \smat{a_1 & & * \\ & \ddots & \\ \nul & & a_k } \right]\,:\, \text{for all $j$,}\;
a_j \in \GL(d(c_j))\,, * \mbox{ arbitrary}\right\} \label{eq-Fa}\,.
\end{align}

To counteract all the rotations let
\begin{equation*}
 \widehat \Rr= \pmat{U_{c_1} & & \nul \\ & \ddots & \\ \nul & & U_{c_k}}\,\in\,{\rm U}(d).
\end{equation*}

\begin{theo} \label{th-flag}
Let $\Tt_0$ be as in \eqref{eq-T0-diag} and let $[\Ff_0] \in \FF^a$ be represented in the form
as described in \eqref{eq-Fa}.
Furthermore let $\Ff_{\lambda,n} = \widehat\Rr^{-n} \Tt_{\lambda,n} \cdots \Tt_{1,\lambda} \Ff_0$.

Then, for fixed $t>0$ and $n\to\infty$ we have $[\Ff_{1/\sqrt{n},\lfloor tn \rfloor}] \Longrightarrow [\Ff_t]$ in law
with
\begin{equation*}
 \Ff_t = \pmat{\Lambda^{(c_1)}_t a_1 & & \nul \\ & \ddots & \\ \nul & & \Lambda^{(c_k)}_t a_k}\,.
\end{equation*}
Here, $\Lambda^{(c_j)}_t$ are the correlated processes for the different magnitudes $c_j$ of eigenvalues of $\Tt_0$
whose correlations are described in Theorem~\ref{theo-SDE2}.
Note that $[\Ff_t] \in \FF^s$.
\end{theo}

\begin{rem}
  If $\Tt_0$ can not be brought into the structure as in \eqref{eq-T0-diag} in general then one still obtains the SDE
  limits on the Grassmannians $G(p,d)$ for $d_2<p\leq d_2+d_1$ as in the proof. $G(p,d)$ denotes the space of $p$-dimensional subspaces of $\CC^d$.
\end{rem}

\begin{coro} \label{cor-flag} Let $\GG\subset \GL(d,\CC)$ be an algebraic group such that the quotient $\GL(d,\CC) / \GG$ is compact. Further more, let $\Tt_0$ be as above.
Then, for $\Xx_0$ in some stable manifold,
$$
\widehat\Rr^{-n} \Tt_{1/\sqrt{n},} \cdots \Tt_{1/\sqrt{n},1}\,\Xx_0\,/\,\GG\;\Longrightarrow\;\Xx_t\,/\,\GG
$$
where $\Xx_t$ satisfies some SDE.
\end{coro}

\begin{proof}
 Under the conditions, $\GL(d,\CC)\,/\,\GG$ is a complex compact algebraic variety and hence a complete algebraic variety. 
 Thus, $\GG$ is a parabolic subgroup which contains a Borel subgroup. We therefore find some $\Gg \in\GL(d,\CC)$ such that the isomorphic conjugate group $\GG'=\Gg\,\GG\,\Gg^{-1}$ 
 contains the Borel subgroup $\Delta(d)$,\,i.e.
 $\Delta(d)\subset \GG'$.
 Therefore, $\GL(d,\CC)\,/\,\GG'$ is itself a quotient of the flag manifold and with $\Ff_0$ and $\Ff_{\lambda,n}$ as above we find
 $\Ff_{1/\sqrt{n},\lfloor nt \rfloor}\,/\,\GG'\,\Rightarrow\, \Ff_t\,/\,\GG'$.
 On easily notices that $M\sim_{\GG'} M'$ if and only if $ M \Gg\,\sim_\GG\, M' \Gg$, and therefore,
 $\Ff_{1/\sqrt{n},\lfloor nt \rfloor}\Gg\,/\,\GG\,\Rightarrow\, \Ff_t \Gg \,/\,\GG$. Hence, we obtain the limiting process when choosing $\Xx_0=\Ff_0 \Gg$ with $\Ff_0$ in the form as above.
\end{proof}

\begin{proof}[Proof of Theorem~\ref{th-flag}.]
As $[\Ff_0]\in\FF^a$ we can represent it by
$$\Ff_0= \smat{a_1 & & * \\ & \ddots & \\ \nul & & a_k }\,,\quad \text{for all $j$,}\;
a_j \in \GL(d(c_j))\,, * \mbox{ arbitrary}\,.$$

Let $G(p,d)$ denote the Grassmannian manifold of $p$-dimensional subspaces of $\CC^d$.
Note that $\Ff^{(p)}\in G(p,d)$.
As $\FF$ can be seen as a submanifold of $\prod_{p=1}^d G(p,d)$ it will be sufficient to
prove $\Ff^{(p)} \Rightarrow \Ff_t^{(p)}$ in $G(p,d)$ jointly for any (fixed) $p$.

As the action of $\Tt$ and $c\Tt$ on $\FF$ or $G(p,d)$ is the same, we may for fixed $p$ scale
the matrices such that $d_2<p\leq d_2+d_1$ in the sense of the definitions of $d_1, d_2$ in the Section~\ref{sub-results}
(Note that this basically means $c_j=1$ for some $j$, $d_2=d(c_1)+d(c_2)+\ldots +d(c_{j-1})$ and $d_1=d(c_j)$.)
Now for $\Ff_1, \Ff_2\in\GL(d,\CC)$ one finds that
\begin{equation}
\Ff_1^{(p)}=\Ff_2^{(p)}\,\quad \text{if and only if} \quad
\Ff_1 = \Ff_2 \pmat{M_1 & \nul \\ * & M_2}\;,\quad M_2\in\GL(p),\;M_1\in\GL(d-p)\;.
\end{equation}
Using blocks of size $d_0+d_1$ and $d_2$ and representing $[\Ff_0]\in\FF^a$ as above we find
\begin{equation}
 \Ff_0=\pmat{ A_0 & B_0 \\ \nul & D_0}\qtx{with} A_0=\pmat{ a_{00} & a_{01} \\ \nul & a_{11}}
\end{equation}
where $D_0$,  $a_{00}$ and $a_{11}$ are invertible.
Note that in fact $a_{11}=a_j$ for some $j$ as in the notations above and that
$a_{00}$ contains the $a_k$ for $k>j$ and $D_0$ contains the $a_k$ for $k< j$.
So we can choose $\Xx_0=\Ff_0$ and consider the processes $\Xx_{\lambda,n}$ as above.
Then clearly $\Ff_{\lambda,n}^{(p)} = \Xx_{\lambda,n}^{(p)}$ and in terms of representatives in $G(p,d)$
they are equivalent to
\begin{equation}\label{eq-p-conj1}
 \pmat{A_{\lambda,n} & B_{\lambda,n} \\ C_{\lambda,n} & D_{\lambda,n}}
 \pmat{\one & \nul \\ -D_{\lambda,n}^{-1} C_{\lambda,n} & D_{\lambda,n}^{-1}} =
 \pmat{X_{\lambda,n} & Z_{\lambda,n} \\ \nul & \one_{d_2}}\,.
\end{equation}
Note that from the proof of Theorem~\ref{theo-SDE} the inverse $D_{\lambda,n}^{-1}$ exists for small $\lambda$ (with sufficiently high probability) 
and therefore, as we consider invertible matrices here, we also find that $X_{\lambda,n}$ is invertible.
As $X_{\frac1{\sqrt{n}},\lfloor tn \rfloor}\smat{\nul \\ \one_{d_1}}\Rightarrow \smat{\nul \\ \Lambda_t D_0}$ with $\Lambda_t$ invertible, we find for $n\sim \lambda^{-2}$ and large $n$ 
that $\smat{ \begin{smallmatrix} \one_{d_0} \\ \nul \end{smallmatrix} & X_{\lambda,n} \smat{\nul \\ \one_{d_1} }}$ is invertible.
Hence, the right hand side of \eqref{eq-p-conj1} represents the same $p$-dimensional subspace as
\begin{equation}
 \pmat{X_{\lambda,n} & Z_{\lambda,n} \\ \nul & \one_{d_2}}
 \pmat{ X_{\lambda_n}^{-1} \smat{ \one \\ \nul} & \begin{smallmatrix} \nul \\ \one \end{smallmatrix} & \nul
\\ \nul & \nul &  \one} =
\pmat{\begin{smallmatrix} \one \\ \nul \end{smallmatrix} & X_{\lambda,n}\smat{\nul \\ \one } & Z_{\lambda,n} \\ \nul & \nul & \one_{d_2}}
\end{equation}
Therefore by Theorem~\ref{theo-SDE} we find
\begin{equation}
 \Ff_{\frac1{\sqrt{n}},\lfloor tn \rfloor}^{(p)} \quad\Longrightarrow\quad
 \pmat{\one_{d_0} & &  \\  & \Lambda_t a_{11} &  \\  &  & \one_{d_2}}^{(p)}\;=\;
 \Ff_t^{(p)}\;.
\end{equation}
The last equation is easy to see if one realizes that the last $p$ column vectors end somewhere inside the $a_{11}$ term and therefore span indeed the same $p$-dimensional subspace as $\Ff_t$.
Clearly, looking at this convergence jointly in $p$ we obtain the correlations as in Theorem~\ref{theo-SDE2}.
\end{proof}


\section{Jordan blocks, critical scalings\label{sub-Jordan} and application at band edges}

Without loss of generality we will focus on the eigenvalues of size $1$ of $\Tt_0$.
Let us introduce the notation $J_k$ for the standard $k\times k$ Jordan block with eigenvalue $1$,
and $N_k$ for the standard Jordan block with eigenvalue $0$, i.e.
\begin{equation*}
 J_k= \pmat{1 & 1 & & \nul\;\; \\  & \ddots & \ddots & \\  & & \ddots & 1 \\ \nul
 & & & 1} = \one + N_k
\end{equation*}
If a Jordan block of the form $e^{i\theta} J_k$ appears in (a possible conjugation of) $\Tt_0$
then we will do a $\lambda$-dependent conjugation.
This trick was already used in \cite{SS1} to analyze the Lyapunov exponent and density of states at a bandedge for a
one-dimensional Schr\"odinger operator. The main point is the following observation.
Define the $\lambda$-dependent, diagonal $k\times k$ matrices
\begin{equation*}
 S_{\lambda,\alpha,k}=\diag(1,\lambda^\alpha,\ldots,\lambda^{(k-1)\alpha})\,
\end{equation*}
then
\begin{equation}\label{eq-conj1}
 S_{\lambda,\alpha,k}^{-1} J_k S_{\lambda,\alpha,k} =
 \one_k + \lambda^\alpha N_k\,.
\end{equation}

Now using blocks of sizes $d_0, d_1, d_2$ as before let
\begin{equation}\label{eq-1jordan}
\Tt_0=\pmat{\Gamma_0 \\ & e^{i\theta} J_{d_1} \\ & & \Gamma_2^{-1}}\;,\quad
\Rr=\pmat{\one \\ & e^{i\theta} \one \\ & & \one}\;,\quad
\Ss_{\lambda,\alpha} = \pmat{\one_{d_0} & \\ & S_{\lambda,\alpha,d_1} \\ & & \one_{d_2}}
\end{equation}
with $\Gamma_1$ and $\Gamma_2$ having spectral radius smaller than one as before.
Conjugating $\Tt_{\lambda,n}$ by $\Ss_{\lambda,\alpha}$ will give a new drift term of order
$\lambda^{\alpha}$ coming from \eqref{eq-conj1}, but it also brings a diffusion term of order
$\lambda^{1-(d_1-1)\alpha}$ from conjugating
$\lambda  \Vv_{\lambda,n}$.
The diffusion has thus order $\lambda^{2-2(d_1-1)\alpha}$ and the most interesting SDE limit arises from balancing
the new drift term and the diffusion term, i.e. $\alpha=2-2(d_1-1)\alpha$, leading to
$\alpha=\alpha(d_1)= 2/(2d_1-1)$. For smaller $\alpha$, the drift term dominates and for larger $\alpha$, the diffusion term dominates.

In fact, only the lower left corner entry
\footnote{If the variance of that entry
happens to be identically zero (no randomness) or of lower order in $\lambda$,
then the diffusion term is of order
$\lambda^{1-(k-2)\alpha}$ (or lower again). This may lead to other interesting scalings as for smaller Jordan blocks.}
of the middle $d_1\times d_1$ block of
$\Ss_{\lambda,\alpha}^{-1} \Vv_{\lambda,n} \Ss_{\lambda,\alpha}$ will be of order
$\lambda^{\alpha/2}$, all other terms from the conjugation will be at least of order $\lambda^{3\alpha/2}$.
Hence, for the case as in \eqref{eq-1jordan} we find
\begin{equation*}
 \Ss_{\lambda,\alpha(d_1)}^{-1} \Tt_{\lambda,n} \Ss_{\lambda,\alpha(d_1)} =
 \widehat \Tt_0 + \lambda^{1/(2d_1-1)} \widehat \Vv_{n} + \lambda^{2/(2d_1-1)} \Nn
 + \lambda^{3/(2d_1-1)} \wt \Vv_{\lambda,n}
\end{equation*}
where
\begin{equation}\label{eq-2jordan}
\wh \Tt_0 = \pmat{\Gamma_0 \\ & e^{i\theta} \one \\ & & \Gamma_2^{-1}}\,,\quad
\Nn=\pmat{\nul \\ & e^{i\theta} N_{d_1} \\ & & \nul}\,,\quad
 \wh\Vv_{n} = \pmat{\nul \\ & V_{11,n} \\ & & \nul}\,.
\end{equation}
Furthermore,  $V_{11,n}$ has only one entry $v_n$ in the lower left corner, and the $v_n$ are i.i.d.
random variables with mean zero,
\begin{equation*}
 V_{11,n}=\pmat{ \nul & \nul_{(d_1-1)\times (d_1-1)}\\ v_n & \nul}\,.
\end{equation*}

Therefore, application of Theorem~\ref{theo-SDE} gives an SDE limit in the scaling
$\lambda^{\alpha_k} n = \lambda^{2/(2k-1)} n=t$:

\begin{theo}
 Let $\Tt_{\lambda,n}$ be given as in \eqref{eq-def-Tt} and let the assumptions as on page \pageref{assumption}
 and \eqref{eq-1jordan} be satisfied. Moreover let $\Ss_{\lambda,\alpha}$ be defined as above with
 $\alpha=\alpha(d_1)=2/(2d_1-1)$. Let
 \begin{equation*}
 \Xx_{\lambda,n} = \Rr^{-n} \Ss_{\lambda,\alpha}^{-1} \Tt_{\lambda,n} \cdots \Tt_{\lambda_1} \Ss_{\lambda,\alpha}
 \Xx_0
 \end{equation*}
 with $\Xx_0$ as before and let $X_{\lambda,n}$ be the corresponding Schur complement as before.
 Then
 \begin{equation*}
  X_{\frac1{\sqrt{n^{2d_1-1}}},\lfloor nt \rfloor}\quad\Longrightarrow\quad X_t =
  \pmat{\nul \\ & \Lambda_t} X_0
 \end{equation*}
\begin{equation}\label{e:jordan}
 d\Lambda_t = N_{d_1} \Lambda_t dt \;+\; \pmat{ \nul & \nul \\ dB_t & \nul} \Lambda_t\;,\qquad \Lambda_0=\one
\end{equation}
where $B_t$ is a complex Brownian motion with covariances
\begin{equation*}
 \EE(B_t^2) = e^{-2i\theta} \EE(v_n^2) \;,\qquad
 \EE(|B_t|^2) = \EE (|v_n|^2)\;.
\end{equation*}
\end{theo}

Note that for a vector $x(t)= \Lambda_t x(0)$ equation \eqref{e:jordan} is equivalent to
\begin{equation}\label{e:jordan2}
 x_1^{(d_1)}= x_1B'
 \end{equation}
and $x_{j+1}=x_1^{(j)}$, the  $j$th derivative of $x_1$, and $B'$ is the (distributional) derivative of the Brownian motion term.

\begin{rem} 
The original drift term coming from $\lambda^2 \Ww_\lambda$ is of too low order after the conjugation with $\Ss_{\lambda,\alpha}$ to matter in the limit.
If one wants an additional drift term in \eqref{e:jordan2} on the right hand side coming from an added term $\lambda^\beta \Ww$ then 
the conjugation $\Ss_{\lambda,\alpha}^{-1} \lambda^\beta \Ww \Ss_{\lambda,\alpha}$ needs to produce a
term of order $\lambda^{\frac{2}{2d_1-1}}=\lambda^\alpha$. If $\Ww$ is not zero in the lower left corner of the corresponding $d_1\times d_1$ block for the SDE limit, then one needs
$\beta-(d_1-1)\alpha=\alpha$, i.e. $\beta=d_1\alpha= 2d_1 / (2d_1-1)$.
\end{rem}

Jordan blocks do appear at so-called band-edges for transfer matrices of one-dimensional random Schr\"odinger operators with some finite range hopping.
Similar as in Section~\ref{sub-RS} consider the random family of random real symmetric matrices $H^{(d)}_{\lambda,n}$ acting on 
$\CC^n\,\ni \psi=(\psi_1,\ldots,\psi_n)$ given by
\begin{equation}
 (H^{(d)}_{\lambda,n} \psi)_k = \sum_{j=0}^{2d} (-1)^{j} \binom{2d}{j} \psi_{k-d+j} \,+\,\lambda v_k\,\psi_k\;,
\end{equation}
where $\psi_j=0$ for $j<1$ and $j>n$. We may sometimes drop the index $n$.
The $v_k$ are independent, identically distributed real random variables with variance $\EE(v_k^2)=1$. Note that for $d=1$ this operator corresponds to \eqref{eq-def-H} with $A=-2$.
The eigenvalue equation $H^{(d)}_\lambda \psi = E\psi$ can be rewritten as
\begin{equation*}
\vec{\psi}_k = (T+(E-\lambda v_k)\,S)\vec{\psi}_{k-1},\qtx{where}{\vec{\psi}_k}=(\psi_{k+d},\psi_{k+d-1},\ldots,\psi_{k-d+1})^\top
\end{equation*}
 and $S$ and $T$ are $2d\times 2d$ matrices given by:
$S_{1,d}=1$ and all other entries of $S$ are zero;
 $T_{1,k}=(-1)^{k+1} \binom{2d}{2d-k}$, $T_{j,j-1} =1$ for $j\geq 2$ and all other entries of $T$ are zero, i.e.
\begin{equation*}
 T = \pmat{\binom{2d}{2d-1} &  \cdots & \binom{2d}{1} & - \binom{2d}{0} \\ 1 & & & 0 \\ 
 & \ddots & & \vdots\\ & & 1 & 0}\;.
\end{equation*}
For $E=0$ and $\lambda=0$ the transfer matrix $T$ is equivalent to a 
Jordan block\footnote{In fact $E=0$ is at the edge of the spectrum of the operator $H^{(d)}_0$ in the limit $n\to \infty$; it is the upper edge for $d$ odd and the lower edge for $d$ even.}  of maximum size for the eigenvalue $1$. 
In order to bring it into the Jordan form, let us define the Pascal-triangle type
matrix $M$ by $M_{jk}=\binom{2d-j}{k-1}$ for $k+j\leq 2d+1$ and zero for all other entries, then one has $M^{-1}_{jk}=(-1)^{j+k} \binom{j-1}{2d-k} $ for $k+j\geq 2d+1$ and all other entries zero, i.e.
\begin{equation*}
 M=\pmat{1 & 2d-1 & \binom{2d-1}{2} & \cdots &  1\\ 1 & 2d-2 & \cdots & 1 \\ 1 & \cdots & 1 \\ \vdots & \iddots \\ 1}\;,\quad
 M^{-1}=\pmat{ & & & & 1 \\ & & & 1 & -1 \\ & & 1 & -2 & 1 \\ & \iddots & \cdots & \cdots & \vdots \\ 1 & \cdots & -\binom{2d-1}{2} & 2d-1 & -1}\;.
\end{equation*}
Then some calculation shows $M^{-1} T M = J_{2d}$ where $J_{2d}$ is the Jordan matrix as defined above. For the conjugation of the whole transfer matrix $T+E-\lambda v_k)S$ we also need to calculate
$M^{-1} S M$. Its entries are given by
$(M^{-1} S M)_{j,k}=M^{-1}_{j,1} S_{1,d} M_{d,k}$ which is only not zero if $j=2d$ and $k\leq d+1$ in which case $(M^{-1} S M)_{2d,k}=\binom{d}{k-1}$, i.e.
\begin{equation*}
 M^{-1} S M = \pmat{0 &  & \cdots & & & \cdots & & 0 \\ \vdots & & & & & & & \vdots \\ 0 & & \cdots & & & \cdots & & 0 \\ 1 & d & \binom{d}{2} & \cdots & 1 & 0 & \cdots & 0}\;.
\end{equation*}
In particular the lower left corner has the entry $1$.
As above let $\alpha=\frac{2}{4d-1}$ 
and as in the remark scale energy differences by $E=\epsilon \lambda^{2d\alpha}$ to obtain
\begin{equation*}
 \Tt_{\lambda,k} := \Ss_{\lambda,\alpha}^{-1} M^{-1} (T+(\lambda^{2d\alpha}\epsilon-\lambda v_k)S)M\,\Ss_{\lambda,\alpha}\,=\,
 \one-\lambda^{\frac{\alpha}{2}} v_k\,Q  + \lambda^{\alpha} \left[ \Nn_{2d}+\epsilon Q\right]+\Oo(\lambda^{\frac{3\alpha}{2}})\;.
\end{equation*}
Then, for any vector $x\in\CC^{2d}$ we find
\begin{equation*}
 \Tt_{n^{-1/\alpha},\lfloor nt \rfloor}\,x\;\Longrightarrow\; x(t)\qtx{with}
x(0)=x,\quad x_1^{(2d)} = x_1  B' + \epsilon x_1\;,\quad x_{j+1}=x_1^{(j)}
\end{equation*}
where $B'$ is the distributional derivative of a standard, real, one-dimensional Brownian motion.

Following the arguments of \cite{KVV} or the arguments of the proof of Theorem~\ref{theo-EV} one could show that
(along suitable subsequences so that the boundary conditions converge) the eigenvalue process of $n^{2d} H^{(d)}_{n^{-1/\alpha},n}$ with $\alpha=2/(4d-1)$ converges to the process of eigenvalues of the random operator
$$
\partial_x^{2d} - B'
$$
acting on the interval $[0,1]$ with appropriate boundary conditions.
For periodic boundary conditions this is a generalization of the random Hill operator (at $d=1$).

\end{document}